\documentclass[10pt]{article}

\usepackage{amsthm}
\usepackage{fullpage}

\usepackage[pdftex]{graphicx}
\graphicspath{{figs/}}

\usepackage[cmex10]{amsmath}

\usepackage{algorithm}
\usepackage{algorithmicx}
\usepackage{algpseudocode}

\usepackage{array}

 \usepackage[tight,footnotesize]{subfigure}

\usepackage{xspace}
\usepackage{cases}
\usepackage{ifthen}
\usepackage{amsfonts}

\begin{document}
%
\title{Efficient and Robust Allocation Algorithms in Clouds under Memory Constraints}


\author{Olivier Beaumont\\
    Inria\\
    Bordeaux, France\\
    Olivier.Beaumont@inria.fr
  \and
Lionel Eyraud-Dubois\\   Inria\\
    Bordeaux, France\\
    Lionel.Eyraud-Dubois@inria.fr
  \and
Paul Renaud-Goud\\
    Inria\\
    Bordeaux, France\\
    Paul.Renaud-Goud@inria.fr
}


%


\maketitle

\begin{abstract}
  We consider robust resource allocation of services in Clouds. More specifically, we consider the
  case of a large public or private Cloud platform that runs a relatively small set of
  large and independent services. These services are characterized by their demand along several dimensions (CPU, memory,$\dots$) and by their quality of service requirements, that
  have been defined through an SLA in the case of a public Cloud or fixed by the
  administrator in the case of a private Cloud. This quality of service defines the
  required robustness of the service, by setting an upper limit on the probability that
  the provider fails to allocate the required quantity of resources. This maximum probability of failure can be transparently turned into a pair $(price,penalty)$. Failures can
  indeed hit the platform, and resilience is provided through service replication.

  Our contribution is two-fold. Firstly, we propose a resource allocation strategy
  whose complexity is logarithmic in the number of resources, what makes it very efficient for large platforms. Secondly, we propose an
  efficient algorithm based on rare events detection techniques in order to estimate
  the robustness of an allocation, a problem that has been proven to be
  \#P-complete. Finally, we provide an analysis of the proposed strategy through an extensive set of simulations, both in terms of the overall number of allocated resources and in terms of time necessary to compute the allocation.
\end{abstract}


%

\newcommand{\ema}[1]{\ensuremath{#1}\xspace}
\newcommand{\inter}[2][1]{\{ #1 , \dots , #2 \}}
\newcommand{\drond}[2]{\frac{\partial #1}{\partial #2}}
\newcommand{\eps}{\ema{\varepsilon}}
\newcommand{\ceil}[1]{\ema{\left\lceil #1 \right\rceil}}
\newcommand{\proba}{\ema{\mathbb{P}}}
\newcommand{\pro}[1]{\ema{\proba\left( #1\right)}}
\newcommand{\pfail}[1]{\ema{\proba_{#1}^{\mathrm{(fail)}}}}
\newcommand{\bino}[2]{\ema{\mathcal{B}\left( #1,#2\right)}}
\newcommand{\argmin}{\mathrm{argmin}}
\newcommand{\cardi}[1]{\ema{\mathit{card} \left\{ #1 \right\}}}

\newcommand{\ie}{{\it i.e.}\xspace}
\newcommand{\aprio}{{\it a priori}\xspace}
\newcommand{\para}[1]{\vspace*{.2cm} {\bf \noindent #1}}
\newtheorem{theorem}{Theorem}
\newtheorem{lemma}{Lemma}
\newtheorem{remark}{Remark}

\newcommand{\homo}{{\it homogeneous}\xspace}
\newcommand{\napp}{\textsc{No-Approx}\xspace}
\newcommand{\norapp}{\textsc{Normal-Approx}\xspace}

\newcommand{\mach}[1]{\ema{\mathcal{M}_{#1}}}
\newcommand{\nm}{\ema{m}}
\newcommand{\fail}{\ema{f}}
\newcommand{\isal}[1]{\ema{\mathit{is\_alive}_{#1}}}
\newcommand{\serv}[1]{\ema{\mathcal{S}_{#1}}}
\newcommand{\ns}{\ema{\mathit{ns}}}
\newcommand{\demand}[1]{\ema{d_{#1}}}
\newcommand{\rel}[1]{\ema{r_{#1}}}
\newcommand{\zr}[1]{\ema{z_{\rel{#1}}}}
\newcommand{\conf}[1]{\ema{\mathcal{C}_{#1}}}
\newcommand{\nc}[1]{\ema{\mathit{nc}_{#1}}}

\newcommand{\mem}[1][]{\ema{M_{#1}}}
\newcommand{\cpu}[1][]{\ema{C_{#1}}}

\newcommand{\pen}[1]{\ema{B_{#1}}}
\newcommand{\dem}[1]{\ema{K_{#1}}}
\newcommand{\allh}[1]{\ema{A_{#1}}}
\newcommand{\all}[2]{\allh{#1,#2}}
\newcommand{\allt}[2]{\ema{a_{\ifthenelse{\equal{#1}{}}{}{#1,}#2}}}
\newcommand{\alla}[2]{\ema{s_{#1,#2}}}
\newcommand{\allf}[2]{\ema{x_{#1,#2}}}
\newcommand{\alcp}[1]{\ema{\mathit{Alive\_cpu}_{#1}}}
\newcommand{\nb}[1]{\ema{n_{#1}}}
\newcommand{\nbe}[1]{\ema{\nb{#1}^{*}}}
\newcommand{\imi}{\ema{i^{-}}}
\newcommand{\ima}{\ema{i^{+}}}
\newcommand{\lam}[1]{\ema{\lambda_{#1}}}
\newcommand{\pd}{\ema{X}}
\newcommand{\pdi}[1]{\ema{\pd_{#1}}}
\newcommand{\pdmi}{\ema{\pd^{(\min)}}}
\newcommand{\pdma}{\ema{\pd^{(\max)}}}

\newcommand{\demt}[2]{\ema{\demand{#1}^{T_{#2}}}}
\newcommand{\nt}[1]{\ema{\mathit{nt}_{#1}}}
\newcommand{\sasi}{\ema{N}}
\newcommand{\rvbino}[1]{\ema{\mathit{Alive}_{#1}}}
\newcommand{\valbino}[2]{\ema{X_{#1}^{\ifthenelse{\equal{#2}{}}{}{(#2)}}}}
\newcommand{\sumall}[1]{\ema{\mathit{sum\_al}\left( #1 \right)}}

\newcommand{\maxshort}[1]{%
  \mathrm{max~} #1 \mathrm{~s.t.~}
}
\newcommand{\minshort}[1]{%
  \mathrm{min~} #1 \mathrm{~s.t.~}
}
\newcommand{\minvshort}[1]{%
  \begin{array}{c}
    \mathrm{min}\\ #1 \\\mathrm{s.t.}
  \end{array}
}
\newcommand{\minlong}[1]{%
  \mathrm{minimize~~} #1 \mathrm{~~such~that~~}
}
\newcommand{\buildproblemn}[3]{%
  \renewcommand\arraystretch{1.3}
  \begin{equation}
    #1 \quad \left\{
      \begin{array}{*1{>{\displaystyle}l}} #2 \end{array} \right.
    #3
  \end{equation}
  \renewcommand\arraystretch{1}}
\newcommand{\buildproblem}[2]{%
  \renewcommand\arraystretch{1.3}
  \begin{equation*}
    #1 \quad \left\{
      \begin{array}{*1{>{\displaystyle}l}} #2 \end{array} \right.
  \end{equation*}
  \renewcommand\arraystretch{1}}

\renewcommand{\algorithmicrequire}{\textbf{Input:}}
\renewcommand{\algorithmicensure}{\textbf{Output:}}
\newcommand{\dcur}{d\_\mathit{cur}}
\newcommand{\minmach}{\mathit{min\_mach}}
\newcommand{\thres}{\mathit{thres}}
\newcommand{\currel}{\mathit{cur\_rel}}
\newcommand{\curthr}{\mathit{cur\_thres}}

\section{Introduction}

Recently, there has been a dramatic change in both the platforms and the
applications used in parallel processing. On the one hand, there has been a dramatic scale change,
that is expected to continue both in data centers and in exascale
machines. On the other hand, a dramatic
simplification change has also occurred in the application models and scheduling
algorithms. On the application side, many large scale applications are
expressed as (sequences of) independent tasks, such as MapReduce~\cite{shih2010performance,dean2008mapreduce,zaharia2008improving}
applications or even run as independent services handling requests.

In fact, the main reason behind this paradigm shift is not related to
scale but rather to unpredictability. First of all, estimating the
duration of a task or the time of a data transfer is extremely
difficult, because of NUMA effects, shared platforms, complicating
network topologies and the number of concurrent
computations/transfers. Moreover, given the number of involved
resources, failures are expected to happen at a frequency such that
robustness to failures is a crucial issue for large scale applications
running on Cloud platforms. In this context, the cost of purely
runtime solutions, agnostic to the application and based either on
checkpointing strategies~\cite{bougeret2011checkpointing,cappello2010checkpointing,bouteiller2013multi} or application
replication \cite{wang2009improving,ferreira2011evaluating} is expected to be large, and there
is a clear interest for application-level solutions that take the inner
structure of the application to enforce fault-tolerance.

In this paper, we will consider reliability issues in a very simple
context, although representative of many Cloud applications. More
specifically, we will consider the problems that arise when allocating
independent services running as Virtual Machines (VMs) onto Physical
Machines (PMs) in a Cloud Computing
platform~\cite{zhang2010cloud,armbrust2009above}. The platforms that
we target have a few crucial properties. First, we assume that the
platform itself is very large, in terms of number of physical
machines. Secondly, we assume that the set of services
running on this platform is relatively small, and that each
service requires a large number of resources. Therefore, our
assumptions corresponds well to a datacenter or a large private Cloud
such as those presented in~\cite{CirneInvited}, but not at all to a
Cloud such as Amazon EC2~\cite{amazon} running a huge number of small
applications. 

In the static case, mapping VMs with heterogeneous computing demands
onto PMs with capacities is amenable to a multi-dimensional
bin-packing problem (each dimension corresponding to a different kind of resource, memory, CPU, disk, bandwidth,$\dots$). Indeed, in this context, on the Cloud
administrator side, each physical machine
comes with its computing capacity (\ie the number of flops it can
process during one time-unit), its disk capacity (\ie the number of
bytes it can read/write during one time-unit), its network capacity
(\ie the number of bytes it can send/receive during one time-unit),
its memory capacity (given that each VM comes with its complete
software stack) and its failure rate (\ie the probability that the
machine will fail during the next time period).  On the client side,
each service comes with its requirement along the same dimensions
(memory, CPU, disk and network footprints) and a reliability demand
that has been negociated through an SLA~\cite{CirneInvited}.

In order to deal with resource allocation
problems in Clouds, several sophisticated techniques have been developed in
order to optimally allocate VMs onto PMs, either to achieve good
load-balancing~\cite{van2009sla,calheiros2009heuristic,christopherTPDS} or
to minimize energy
consumption~\cite{berl2010energy,beloglazov2010energy}. Most of the
works in this domain have therefore focused on designing
offline~\cite{GareyJohnson} and
online~\cite{EpsteinvanStee07binpackresaumg,Hochbaum97} solutions of
Bin Packing variants.

In this paper, we propose to reformulate these heterogeneous resource
allocation problems in a form that takes advantage on the assumptions
we made on the platform and the characteristics of VMs. This set of assumption is crucial for the algorithms we propose. In this perspective, since we assume that the number of processing resources \nm is very large, so that we will focus on resource allocation algorithms whose complexity is low (in practice logarithmic) in \nm. We also assume that each VM comes with its full software stack, so that the number of different services $K_{\max}$ that can actually run on the platform is very small, and will be treated as a small constant. We also assume that the number of services \ns is also small. Typically, we will propose algorithms that will rely on the partial enumeration of the possible configurations of a set of PMs (\ie the set of applications they run) and more precisely on column generation techniques~\cite{barnhart1998branch,desrosiers2005primer} in order to solve efficiently the optimization problems.
In Section~\ref{sec.simus}, we will provide a detailed analysis of the resource allocation algorithm that we propose in this paper for a wide set of practical parameters. The cost of the algorithm will be analyzed both in terms of their processing time to find the allocation and in terms of the quality of the computed allocation (\ie required number of resources).

Reliability constraints have received much less attention in the
context of Cloud computing, as underlined by Cirne {\it et
  al.}~\cite{CirneInvited}. Nevertheless, reliability issues have been
addressed in more distributed and less reliable systems such as
Peer-to-Peer networks. In such systems, efficient data sharing is
complicated by erratic node failure, unreliable network connectivity
and limited bandwidth. In this case, data replication can be used to
improve both availability and response time and the question is to
determine where to replicate data in order to meet simultaneously
performance and availability requirements in large-scale
systems~\cite{replication_availability,replication_cirne_1,replication_datagrid,replication_DB,replication_cirne_2}.
Reliability issues have also been addressed by High Performance
Computing community. Indeed, the exascale
community~\cite{dongarra2009international,eesi} underlines the
importance of fault tolerance issues~\cite{cappello2009fault} and
proposed solutions based either on replication
strategies~\cite{ferreira2011evaluating,wang2009improving} or
rollback recovery relying on checkpointing
protocols~\cite{bougeret2011checkpointing,cappello2010checkpointing,bouteiller2013multi}.

This work is a follow-up of~\cite{beaumont:hal-00743524}, where the
question of how to evaluate the reliability of an allocation has been
addressed. One of the main results of~\cite{beaumont:hal-00743524} was
that estimating the reliability of a given allocation was already a
\#P-complete
problem~\cite{valiant1979complexity,provan1983complexity,bodlaender2004note}. In
this paper, we prove that rare event detection
techniques~\cite{algo-rare} developed by the Applied Probability
community are in fact extremely efficient in practice to circumvent
this complexity result. This paper is also a follow-up
of~\cite{beaumont2013hipc}, where asymptotically approximation
algorithms for energy minimization have been proposed in the context
where services were defined by their processing requirement only (and
not their memory requirement). In~\cite{beaumont2013hipc},
approximation techniques to estimate the reliability of an allocation
were based on the use of Chernoff~\cite{chernoff} and
Hoeffding~\cite{hoeffding} bounds. In the present paper, we propose
different techniques based on the approximation of the binomial distribution by
the Gaussian distribution.

The paper is organized as follows. In Section~\ref{sec.frame}, we
present the notations that will be used throughout this paper and we
define the characteristics of both the platform and the services that
are suitable for the techniques we propose. In
Section~\ref{sec.resolution}, we propose an algorithm for
solving the resource allocation problem under reliability
constraints. It relies on a pre-processing phase that is used to
decompose the problem into a reliability problem and a packing
problem. In Section~\ref{proba.est}, we propose a new technique based
on rare event detection techniques to estimate the reliability of an
allocation, and we prove that this technique is very efficient in our
context. At last, we present in Section~\ref{sec.simus} a set of
detailed simulation results that enable to analyze the performance of the algorithm
proposed in Section~\ref{sec.resolution} both in terms of the quality
of returned allocations and processing time. Concluding remarks
are presented in Section~\ref{sec.conclusion}.

\section{Framework}
\label{sec.frame}

\subsection{Platform and services description}

In this paper, we assume the following model. On the one hand, the platform is composed
of \nm homogeneous machines $\mach{1}, \dots, \mach{\nm}$, that have
the same CPU capacity \cpu and the same memory capacity \mem. On the
other hand, we aim at running \ns services $\serv{1}, \dots, \serv{\ns}$,
that come with their CPU and memory requirements. 
In this context, our goal is to minimize the number of used machines, and at to find an allocation of the services onto the machines
such that all packing constraints are fulfilled. Nevertheless, this problem is not equivalent to a 
classical multi-dimensional packing problem.  Indeed, the two
requirements are of a different nature.

On the one hand, services are heterogeneous from a CPU perspective,
hence each service \serv{i} expects that it will be provided a total
computation power of \demand{i} (called demand) among all the
machines. In addition, CPU sharing is modeled in a fluid manner:
on a given machine, the fraction of the total CPU dedicated to a given
service can take any (rational) value.  This expresses the fact that the
sharing between services which are running on a given machine is done
through time multiplexing, whose grain is very fine.

On the other hand, memory requirements are homogeneous among all
services, and memory requirements cannot be partially allocated: running a service
on a machine occupies one unit of memory of this machine, regardless
of the amount of computation power allocated to this service. This
assumption models the fact that most of the memory used by virtual
machines comes from the complete software stack image that needs to be
deployed. On the other hand, the homogeneous assumption is not a strong
requirement, and our algorithms could be modified to account for
heterogeneous memory requirements (at the price of more complex notations). Furthermore, since the complete software stack image is needed,
we assume that the memory capacity of the machines is not very high,
\ie each machine can hold at most 10 services.

\subsection{Failure model}

In this paper, we envision large-scale platforms, which means that
machine failures are not uncommon and need to be taken into
account. Two techniques are usually set up to face those machine
failures: migration and replication. The response time of migrations
may be too high to ensure continuity of the services. Therefore, we concentrate in this
paper on a phase that occurs between two migration and reallocation operations, that are scheduled every $x$ hours. The migration and reallocation strategy is out of the scope of this paper and we rather concentrate on the use of replication in order to provide the
resilience between two migration phases. The SLA defines the robustness properties that the allocation should have.

More specifically, we assume that machine
failures are independent, and that machines are homogeneous also with
regard to failures. We denote \fail the probability that a given
machine fails during the time period between two migration phases.
Because of those failures, we cannot ensure that a service will have
enough computational power at its disposal during the whole time
period. The probability that all machines in the platform fail is
indeed positive. Therefore, in our model each service \serv{i} is also
described with its reliability requirement \rel{i}, which expresses a
constraint: the probability that the service has not enough
computational power (less than its demand \demand{i}) at the end of
the time period must be lower than \rel{i}. 

In this context, replicating a given service of many machines whose failure are independent, it will be possible to achieve any reliability requirement. Our goal in this paper is to do it for all services simultaneously, \ie to enforce that capacity constraints, reliability requirements and 
service demands will be satisfied, while minimizing the number of required machines.

\
subsection{Problem description}

We are now ready to state precisely the problem. Let \all{i}{j} be
the CPU allocated to service \serv{i} on machine \mach{j}, for all $i
\in \inter{\ns}$ and $j \in \inter{\nm}$. For all $j \in \inter{\nm}$,
we denote \isal{j} the random variable which is equal to 1 if machine
\mach{j} is alive at the end of the time period, and 0 otherwise. We
can then define, for all $i \in \inter{\ns}$, the total CPU amount
that is available to service \serv{i} at the end of the time period:
$\alcp{i} = \sum_{j=1}^{m} \isal{j} \times \all{i}{j}$. The problem of
the minimization of the number of used machines can be written as:
\begin{numcases}{\minvshort{\nm}}
  \forall i, \pro{ \alcp{i} < \demand{i}} < \rel{i} \label{eq.init.rel}\\ 
  \forall j,  \cardi{ \all{i}{j} \neq 0 \, ; i \in \inter{\ns}} \leq \mem
  \label{eq.init.mem}\\
  \forall j,  \sum_{i=1}^{\ns} \all{i}{j} \leq \cpu \label{eq.init.cpu}
\end{numcases}

Equations~\eqref{eq.init.mem} and~\eqref{eq.init.cpu} depict the packing
constraints, while Equation~\eqref{eq.init.rel} deals with reliability
requirements.

We will use in this paper two approaches for the estimation of the
reliability requirements. In the \napp model, the reliability
constraint is actually written $\pro{ \alcp{i} < \demand{i}} <
\rel{i}$.

However, as previously stated, given an allocation of one service onto the
machines, deciding whether this allocation fulfills the reliability
constraint or not is a \#P-complete
problem~\cite{beaumont:hal-00743524}; this shows that estimating this
reliability constraint is a hard task.  In~\cite{nous-resilience}, it has been observed that, based on the
approximation of a binomial distribution by a Gaussian distribution,
$\pro{ \alcp{i} < \demand{i}} < \rel{i}$ is approximately equivalent
to
\begin{eqnarray*}
  \sum_{j=1}^{\nm} \all{i}{j} - \pen{i} \sqrt{ \sum_{j=1}^{\nm} \all{i}{j}^2}
  \geq \dem{i}, \mathrm{~~where}\\
  \dem{i}=\frac{\demand{i}}{1-\fail} \quad \mathrm{and} \quad
  \pen{i} = \zr{i} \times \sqrt{\frac{\fail}{1-\fail}}.
\end{eqnarray*}
$\zr{i}$ is a characteristic of normal distributions, and only depends
on $\rel{i}$, therefore it can be tabulated beforehand. In the
following, we will denote this model by the \norapp model.

\medskip

Both packing and fulfilling the reliability constraints are hard
problems on their own, and it is even harder to deal with those two
issues simultaneously.  In the next section, we describe the way we
solve the global problem, by decomposing it into two sub-problems that
are easier to tackle.

\medskip

\section{Problem resolution}
\label{sec.resolution}

We approach the problem through a two-step heuristic. The first step
focuses mainly on reliability issues. The general idea about
reliability is that, for a given service, in order to keep the
replication factor low (and thus reduce the total number of machines
used), the service has to be divided into small slices and distributed
among sufficiently many machines. However, using too many small slices for each service
would break the memory constraints (remember that the memory requirement associated to a service is the same whatever the size of slice, as soon as it is larger than 0). The goal of the first step,
described in Section~\ref{sec.homo}, is thus to find reasonable slice
sizes for each service, by using a relaxed packing formulation which
can be solved optimally.

In a second step, described in Sections
~\ref{sec.cg}, we compute the actual packing of those service slices
onto the machines.  Since the number of different services allocated
to each machine is expected to be low (because of the memory
constraints), we rely on a formulation of the problem based on the partial
enumeration of the possible configurations of machines, and we use
column generation techniques~\cite{barnhart1998branch,desrosiers2005primer} to limit the number of different configurations.

\subsection{Focus on reliability}
\label{sec.homo}

In this section, we describe the first step of our approach: how to
compute allocations that optimize the compromise between reliability
and packing constraints, under both \napp and \norapp models. This is
done by considering a simpler, relaxed formulation of the problem,
that can be solved optimally. We start with the \norapp model.

\subsubsection{\norapp model}

As stated before, in this first phase, we relax the problem by
considering global capacities instead of capacities per machine. Thus
we dispose of a total budget $\nm \mem$ for memory requirements and
$\nm \cpu$ for CPU needs, and use the following formulation:
\begin{numcases}{\minvshort{\nm}}
  \forall i, \sum_{j=1}^{\nm} \all{i}{j} - \pen{i} \sqrt{ \sum_{j=1}^{\nm} \all{i}{j}^2}
  \geq \dem{i} \label{eq.pr.homo.r}\\ 
  \sum_{j=1}^{\nm} \mathit{card} \left\{ \all{i}{j} \neq 0 \, ;
    i \in \inter{\ns} \right\}\leq \nm \mem \label{eq.pr.homo.m} \\
  \sum_{j=1}^{\nm}  \sum_{i=1}^{\ns} \all{i}{j} \leq \nm \cpu \label{eq.pr.homo.c}
\end{numcases}

In the following, we prove that this formulation can be solved
optimally. We define a class of solutions, namely \homo allocations,
in which each service is allocated on a set of machines, with the same
CPU requirement. Formally, an allocation is \homo if for all $i \in
\inter{\ns}$, there exists $\allh{i}$ such that for all $j \in
\inter{\nm}$, either $\all{i}{j} = \allh{i}$ or $\all{i}{j}=0$.

\begin{lemma}
  On the relaxed problem, \homo allocations is a dominant class of
  solutions.
\end{lemma}

\begin{proof}
  Let us assume that there exist $i$, $j_1$ and $j_2 \neq j_1$, such
  that $\all{i}{j_1} > \all{i}{j_2} > 0$. By setting $\all{i}{j_1}' =
  \all{i}{j_2}' = (\all{i}{j_1}+\all{i}{j_2})/2$, we increase the
  left-hand side of Equation~\eqref{eq.pr.homo.r}, and leave unchanged
  the left-hand sides of Equations~\eqref{eq.pr.homo.m}
  and~\eqref{eq.pr.homo.c}. From any solution of the problem, we can
  build another \homo solution, which does not use a larger number of
  machines.
\end{proof}

An \homo allocation is defined by \nb{i}, the number of machines
hosting service \serv{i}, and \allh{i}, the common CPU consumption of
service \serv{i} on each machine it is allocated to. The problem of
finding an optimal \homo allocation can be written as:

\begin{numcases}{\minshort{\nm}}%
  \forall i, \nb{i} \allh{i} - \pen{i} \allh{i} \sqrt{\nb{i}} \geq \dem{i}\label{eq.homo.r}\\ 
  \forall i, \allh{i} \geq 0 \nonumber\\
  \sum_i \nb{i} \leq \nm \mem\nonumber\\
  \sum_i \nb{i} \allh{i}  \leq \nm \cpu\nonumber
\end{numcases}



which can be simplified into
\buildproblemn{\minshort{\nm}}{%
  \forall i, \sqrt{\nb{i}} > \pen{i} \\
  \sum_i \nb{i} \leq \nm \mem \\
  \sum_i \frac{\dem{i}}{1 - \frac{\pen{i}}{\sqrt{\nb{i}}}}  \leq \nm \cpu
}{\label{pb.init}}

In the following, we search for a fractional solution to this problem:
both \nm, the \nb{i}'s and the \allh{i}'s are assumed to be rational
numbers. We begin by formulating two remarks to help solving this
problem.

\bigskip

\begin{remark}
  \label{rem.eq}
  We can restrict to solutions which satisfy the following constraints: 

  \begin{numcases}{}
    \forall i, \sqrt{\nb{i}} > \pen{i} \nonumber\\
    \sum_i \nb{i} = \nm \mem \nonumber\\
    \sum_i \frac{\dem{i}}{1 - \frac{\pen{i}}{\sqrt{\nb{i}}}} = \nm \cpu \nonumber
  \end{numcases}
\end{remark}

\begin{proof}
  For all $i$, $\frac{\dem{i}}{1 - \frac{\pen{i}}{\sqrt{\nb{i}}}}$ is
  a non-increasing function of \nb{i}, thus given a solution of
  Problem~\ref{pb.init}, we can build another solution such that
  $\sum_i \nb{i} = \nm \mem$. Indeed, let $(\nb{1},\dots,\nb{\ns})$ be
  an optimal solution of Problem~\ref{pb.init}, and let $\nb{i}' =
  \nb{i}$ for all $i \in \inter[2]{\ns}$.  Now if we set $\nb{1}' =
  \nm \mem - \sum_{i=2}^\ns \nb{i}$, we have on the one hand
  $\sum_{i=1}^\ns \nb{i}' = \nm \mem$, and on the other hand
  $\sum_{i=1}^\ns \frac{\dem{i}}{1 - \pen{i}/\sqrt{\nb{i}'}} \leq \nm
  \cpu$, since $\nb{1}' \geq \nb{1}$.

  In the following, we only consider such solutions: let
  $(\nb{1},\dots,\nb{\ns})$ be a solution of Problem~\ref{pb.init}
  with \nm machines, and satisfying $\sum_i \nb{i} = \nm \mem$. Let us
  now further assume that, in this solution, $\sum_i \frac{\dem{i}}{1
    - \frac{\pen{i}}{\sqrt{\nb{i}}}} < \nm \cpu$.  We show that such a
  solution is not optimal by exhibiting a valid solution
  $(\nb{1}',\dots,\nb{\ns}')$, which uses $\nm' < \nm$ machines.  We
  set, for all $i \in \inter[2]{\ns}$, $\nb{i}' = \nb{i}$, and we
  define $\nb{1}'$ such that:
  \begin{align*}
    \sqrt{\nb{1}'} = \max \Bigg( & \frac{\sqrt{\nb{1}} + \pen{1}}{2} , \\
    &  \frac{\pen{1}}{1-\frac{2\pen{1}}{\nm \cpu -
        \sum_{i \neq 1} \frac{\dem{i}}{1-\pen{i}/\sqrt{\nb{i}}}
        + \frac{\dem{1}}{1-\pen{1}/\sqrt{\nb{1}}}}}
    \Bigg) .\\
  \end{align*}
  We have firstly $\sqrt{\nb{1}'} \geq (\sqrt{\nb{1}} + \pen{1})/2 > 2
  \pen{1} / 2$, since $\sqrt{\nb{1}} > \pen{1}$.  Furthermore, we
  prove now that $\sqrt{\nb{1}'}<\sqrt{\nb{1}}$.  On the one hand,
  again from $\pen{1} < \sqrt{\nb{1}}$, we obtain $(\sqrt{n_1} +
  \pen{1})/2 < \sqrt{n_1}$. On the other hand, from $\sum_i
  \frac{\dem{i}}{1 - \frac{\pen{i}}{\sqrt{\nb{i}}}} < \nm \cpu$, we
  have
  \[ \frac{\pen{1}}{1-\frac{2\dem{1}}{\nm \cpu -
      \sum_{i \neq 1} \frac{\dem{i}}{1-\pen{i}/\sqrt{\nb{i}}}
      + \frac{\dem{1}}{1-\pen{1}/\sqrt{n_1}}}} < \sqrt{n_1}.\]
  Finally $\sqrt{\nb{1}'}<\sqrt{\nb{1}}$, hence
  $\sum_{i} \nb{i}'  < \nm \mem$.

  The second term of the maximum ensures that:
  \[ \frac{\dem{1}}{1 - \frac{\pen{1}}{\sqrt{\nb{i}'}}} <
  \nm \cpu - \sum_{i \neq 1} \frac{\dem{i}}{1 - \frac{\pen{i}}{\sqrt{\nb{i}}}}. \]

  All together, $(\nb{1}', \nb{2}', \dots, \nb{\ns}')$ satisfies
  \begin{numcases}{}
    \forall i, \sqrt{\nb{i}'} > \pen{i} \nonumber\\ 
    \sum_{i} \nb{i}'  < \nm \mem  \nonumber\\ 
    \sum_{i} \frac{\dem{1}}{1 - \frac{\pen{1}}{\sqrt{\nb{i}'}}} < \nm \cpu  \nonumber
  \end{numcases}
  which implies that $(\nb{1},\dots,\nb{\ns})$ is not an optimal
  solution.

\end{proof}

\bigskip

\begin{remark}
  \label{rem.drond}
  Let us now define $f_1$ and $f_2$ by $f_1\left( \nb{1}, \dots,
    \nb{\ns} \right) = \sum_i \nb{i}$ and $f_2\left( \nb{1}, \dots,
    \nb{\ns} \right) = \sum_i \frac{\dem{i}}{1 -
    \frac{\pen{i}}{\sqrt{\nb{i}}}}$.

  Necessarily, at a solution with minimal \nm, we have:
  \[ \forall i,j \quad \drond{f_2}{\nb{i}} = \drond{f_2}{\nb{j}} .\]
\end{remark}
\begin{proof}
  For a given solution $\nb{} = (\nb{1},\dots,\nb{\ns})$, let us
  assume that there exist $i$ and $j$ such that $\drond{f_2}{\nb{i}} <
  \drond{f_2}{\nb{j}}$. Without loss of generality, we can assume that
  $i<j$. At the first order,
  \begin{multline*}
    f_2\left( \nb{1}, \dots,\nb{i}+\eps,\dots,\nb{j}-\eps,
      \dots, \nb{\ns} \right) =\\
    f_2\left( \nb{1}, \dots, \nb{\ns} \right) + \eps 
    \left( \drond{f_2}{\nb{i}} - \drond{f_2}{\nb{j}} \right) + o(\eps).
  \end{multline*}
  Then there exists $\eps>0$ such that $\nb{j}-\eps> \pen{j}$ and
  $f_2\left( \nb{1}, \dots,\nb{i}+\eps,\dots,\nb{j}-\eps, \dots,
    \nb{\ns} \right) < f_2\left( \nb{1}, \dots,\nb{\ns} \right)$.
  Moreover, in the same way as in Remark~\ref{rem.eq}, we can show
  that there also exists $\eps'$ such that $\nb{j}-\eps-\eps' >
  \pen{j}$ and $f_2\left( \nb{1},
    \dots,\nb{i}+\eps,\dots,\nb{j}-\eps-\eps', \dots, \nb{\ns} \right)
  < f_2\left( \nb{1}, \dots, \nb{\ns} \right) = \nm \cpu$.  By
  remarking that $f_1\left( \nb{1}, \dots,\nb{i}+\eps,\dots,
    \nb{j}-\eps-\eps',\dots, \nb{\ns} \right) < f_1\left( \nb{1},
    \dots, \nb{\ns} \right) = \nm \mem$, we show that $\nb{}$ is not
  an optimal solution.
\end{proof}

\bigskip

Both remarks show that at an optimal solution point, there exists \pd
such that
\[ \forall i \quad -\frac{\pen{i} \dem{i}}{\sqrt{\nb{i}}(\sqrt{\nb{i}} - \pen{i})^2}
= \drond{f_2}{\nb{i}} = \pd.\]

\para{Computing the \nb{i}'s given \pd}

By denoting $x_i = \sqrt{\nb{i}}$, let us consider the following
third-order equation $x_i(x_i - \pen{i})^2 + \frac{\pen{i}
  \dem{i}}{\pd} = 0$.  The derivative is null at $x_i=\pen{i}$ and
$x_i=\pen{i}/3$, and the function tends to $+\infty$ when $x_i
\rightarrow +\infty$. Since we search for $x_i>\pen{i}>0$, we deduce
that for any $\pd < 0$, this equation has an unique solution. Let us
denote $g_i(X)$ the unique value of $\nb{i}$ such that $\sqrt{\nb{i}}$
is a solution to this equation. As $x_i (x_i - \pen{i})^2 \leq x_i^3$,
we know that $x_i \leq \sqrt[3]{-\pen{i} \dem{i}/\pd}$.  We can thus
compute $g_i(X)$ with a binary search inside $]\pen{i},\sqrt[3]{ -
  \pen{i} \dem{i}/\pd}]$, since $x \mapsto x (x - \pen{i})^2 \leq x^3$
is an increasing function in this interval. Incidentally, we note that
for all $i$, $g_i$ is an increasing function of $X$.

\para{Computing \pd}

According to remark~\ref{rem.drond}, for any optimal solution there
exists $X$ such that
\begin{equation}\label{eq.pd}\sum_i g_i(X) = \mem/\cpu \times \sum_i
\frac{\dem{i}}{1 - \frac{\pen{i}}{\sqrt{g_i(X)}}}.
\end{equation}
Since the left-hand side is increasing with \pd, and the right-hand
side is decreasing with \pd, this equation has an unique solution
$\pd^*$ which can be computed by a binary search on \pd. Once $\pd^*$
is known, we can compute the $\nb{i}$ and we are able to derive the
$\allh{i}$'s.  The solution $S^*$ computed this way is the unique
optimal solution: for any optimal solution $S'$, there exists $\pd'$
which satisfies the previous equation. Since this equation has only
one solution, $\pd' = \pd^*$ and $S' = S^*$.

We now show how to compute upper and lower bounds for the binary
search on $X$. As shown previously, we have an obvious upper bound:
$\pd <0$.  We express now a lower bound. Let \nbe{i} be defined, for
all $i$, by

\[ \nbe{i} \geq 0 \mathrm{~~and~~} \nbe{i}=\frac{\mem}{\cpu} \times
\frac{\dem{i}}{1-\frac{\pen{i}}{\sqrt{\nbe{i}}}}.\]

Then $\sqrt{\nbe{i}}$ is a solution of a second-order equation, $
\nbe{i} - \pen{i} \sqrt{\nbe{i}} - \mem\dem{i}/\cpu = 0$, and since
$\nbe{i} \geq 0$, we can compute:
\[ \sqrt{\nbe{i}} = \frac{1}{2} \times \left( \pen{i} +
  \sqrt{\pen{i}^2 + \frac{4\mem\dem{i}}{\cpu}} \right). \] 

Now let
\[ \pdi{i} = - \frac{\pen{i} \dem{i}}%
{\sqrt{\nbe{i}} \left( \sqrt{\nbe{i}} - \pen{i} \right)^2} \quad
\text{and} \quad \imi = \argmin_i \; \pdi{i}. \]

For all~$i$, $\pdi{\imi} \leq \pdi{i}$.  Since $g_i$ is increasing
with $X$, we have $g_i(\pdi{\imi}) \leq g_i(X_i) = \nbe{i}$. Since $x
\mapsto \dem{i}/(1-\pen{i}/\sqrt{x})$ is non-increasing, this implies

\[ \frac{\mem}{\cpu} \frac{\dem{i}}{1-\pen{i}/\sqrt{\nbe{i}}} \leq
\frac{\mem}{\cpu} \frac{\dem{i}}{1-\pen{i}/\sqrt{g_i(\pdi{\imi})}}. \] 

From the definition of \nbe{i}, we can conclude
$$ \sum_i g_i(\pdi{\imi}) \leq \frac{\mem}{\cpu} \times \sum_i
\frac{\dem{i}}{1-\frac{\pen{i}}{\sqrt{g_i(\pdi{\imi})}}}\quad
\Rightarrow \quad \pd^* \geq \pdi{\imi}. $$

With the same line of reasoning, we can refine the upper bound into
$\pd^* \leq \pdi{\ima}$, where $\ima = \mathrm{argmax}_i \; \pdi{i}$,
by showing
\[ \sum_i g_i(\pdi{\ima}) \geq \frac{\mem}{\cpu} \times \sum_i
\frac{\dem{i}}{1-\frac{\pen{i}}{\sqrt{g_i(\pdi{\ima})}}}.\]

\subsubsection{\napp model}

In the previous section, we showed how to compute an optimal solution
to the relaxed problem under the \norapp model, but we have no
guarantee that this solution will meet the reliability constraints
under the \napp model.

Given an homogeneous allocation for a given service \serv{i}, the
amount of alive CPU of \serv{i} follows a binomial law: $\alcp{i} \sim
\allh{i} \times \bino{\nb{i}}{1-\fail}$. We can then rewrite the
reliability constraint, under the \napp model, as $\pro{\allh{i}
  \times \bino{\nb{i}}{1-\fail} < \demand{i}} < \rel{i}$.  This
constraint describes the actual distribution, but since the values
$(\nb{i}, \allh{i})$ have been obtained via an approximation, there is
no guarantee that they will satisfy this constraint.  However, since
the cumulative distribution function of a binomial law can be computed
with a good precision very efficiently~\cite{lib-binomial}, we can
compute $\nb{i}'$, the first integer which meets the constraint. We
can the use equation~\eqref{eq.homo.r} to refine the value of
$\pen{i}$: we compute $\pen{i}'$ so that equation~\eqref{eq.homo.r}
with $\allh{i}$ and $\nb{i}'$ is an equality, so that the
approximation of the \norapp model is closer to the actual
distribution for these given values of $\allh{i}$ and $\nb{i}$.


We compute new \pen{i}'s for all services, and iterate on the
resolution of the previous problem, until we reach a convergence point
where the values of the \pen{i} do not change. In our simulations (see
Section~\ref{sec.simus}), this iterative process converges in at most 10
iterations.

\subsection{Focus on packing}
\label{sec.cg}

\newcommand{\alfull}{\ema{\mathcal{F}}}
\newcommand{\progCG}{\ema{\mathcal{P}}}
\newcommand{\dsknap}{\textsc{Split-Knapsack}\xspace}

In the previous section, we have described how to obtain an optimal
solution to the relaxed problem~\ref{pb.init}, in which \homo
solutions are dominant. In the original problem, packing constraints
are expressed for each machine individually, and the flexibility of
non-homogeneous allocations may make them more efficient. Indeed, an
interesting property of equation~\eqref{eq.pr.homo.r} is that
"splitting'' a service (\ie, dividing an allocated CPU consumption on
several machines instead of one) is always beneficial to the
reliability constraint (because splitting keeps the total sum
constant, and decreases the sum of squares). In this Section, we thus
consider the packing part of the problem, and the reliability issues
are handled by the following constraints: the allocation of service
\serv{i} on any machine $j$ should not exceed \allh{i}, and the total
CPU allocated to \serv{i} should be at least $\nb{i}\allh{i}$. Since
the $(\nb{i}, \allh{i})$ values are such that the homogeneous
allocation satisfies the reliability constraint, the splitting
property stated above ensures that any solution of this packing
problem satisfies the reliability constraint as well.


The other idea in this Section is to make use of the fact that the
number $\mem$ of services which can be hosted on any machine is
low. This implies that the number of different machine configurations
(defined as the set of services allocated to a machine) is not too
high, even if it is of the order of $\ns^M$. We thus formulate the
problem in terms of \emph{configurations} instead of specifying the
allocation on each individual machines. However, exhaustively
considering all possible configurations is only feasible with
extremely low values of \mem (at most 4 or 5). In order to address a
larger variety of cases, we use in this section a standard column
generation method~\cite{barnhart1998branch,desrosiers2005primer} for
bin packing problems.

In this formulation, a configuration \conf{c} is defined by the
fraction \allf{i}{c} of the maximum capacity \allh{i} devoted to
service \serv{i}. According to the constraints stated above,
configuration \conf{c} is \emph{valid} if and only if $\sum_i
\ceil{\allf{i}{c}} \leq \mem$, $\sum_i \allf{i}{c}\allh{i} \leq \cpu$,
and $\forall c, 0\leq \allf{i}{c} \leq 1$.  Furthermore, we only
consider \emph{almost full} configurations, defined as the
configurations in which all services are assigned a capacity either
$0$ or $1$, except at most one. Formally, we restrict to the set
\alfull of valid configurations \conf{c} such that $\cardi{ i, 0 <
  \allf{i}{c} < 1 } \leq 1$.

We now consider the following linear program \progCG, in which there
is one variable \lam{j} for each valid and almost full configuration:

\begin{equation}
  \minshort{\sum_{c\in\alfull}  \lam{c}} \quad \forall i, \sum_{c\in
    \alfull} \lam{c}\allf{i}{c} \geq \nb{i} \label{pb.colgen}
\end{equation}

Despite the high number of variables in this formulation, its simple
structure (and especially the low number of constraints) allows to use
column generation techniques to solve it. The idea is to generate
variables only from a small subset $\alfull'$ of configurations and
solve the problem \progCG on this restricted set of variables. This
results in a sub-optimal solution, because there might exist a
configuration in $\alfull\setminus \alfull'$ whose addition would
improve the solution. Such a variable can be found by writing the dual
of \progCG (the variables in this dual are denoted $p_i$):

\begin{equation*}
  \maxshort{\sum_{i} \nb{i}p_i} \quad \forall c\in\alfull, \sum_{i} \allf{i}{c}p_i \leq 1
\end{equation*}

The sub-optimal solution to \progCG provides a (possibly infeasible)
solution $p_i^*$ to this dual problem. Finding an improving
configuration is equivalent to finding a violated constraint, {\it
  i.e.} a valid configuration \conf{j} such that $\sum_i
\allf{i}{c}p_i^* > 1$. We can thus look for the configuration \conf{j}
which maximizes $\sum_i \allf{i}{c}p_i^*$. This sub-problem is a
knapsack problem, in which at most one item can be split. 

Let us denote this knapsack sub-problem as \dsknap. It can be
formulated as follows: given a set of item sizes $s_i$, item profits
$p_i$, a maximum capacity $C$ and a maximum number of elements $M$,
find a subset $J$ of items with weights $x_i$ such that $\cardi{J}
\leq M$, and $\sum_{i\in J} x_i s_i \leq C$ which maximizes the profit
$\sum_{i\in J} x_i p_i$. We first remark that solutions with at most
one split item are dominant for \dsknap (which justifies that we only
consider almost full valid configurations, \ie configurations with at most one split item). Then, we prove that this
problem is NP-complete, and we propose a pseudo-polynomial dynamic
programming algorithm to solve it. This algorithm can thus be used to
find which configuration to add to a partial solution of \progCG to improve it.
However, for comparison purposes, we also use a Mixed Integer
Programming formulation of this problem which is used in the
experimental evaluation in Section~\ref{sec.simus}.

\begin{remark}
  \label{rem.split.dominant}
  For any instance of \dsknap, there exists an optimal solution with at
  most one split item (\ie, at most one $i\in J$ for which $0 < x_i <
  1$). Furthermore, this split item, if there is one, has the smallest
  $\frac{p_i}{s_i}$ ratio.
\end{remark}
\begin{proof}
  This is a simple exchange argument: let us consider any solution $J$
  with weights $x_i$, and assume by renumbering that $J=\{1, \dots,
  m\}$ and that items are sorted by non-increasing $\frac{p_i}{s_i}$
  ratios. We can construct the following greedy solution: assign
  weight $x'_i = 1$ to the first item, then to the second, until the
  first value $k$ such that $\sum_{i \leq k} s_i > C$, and assign
  weight $x'_k = (C - \sum_{i < k} s_i)/s_k$ to item $k$. It is
  straightforward to see that this greedy solution is valid, splits at
  most one item, and has profit not smaller than the original
  solution.
\end{proof}

\begin{theorem}
  \label{thm.split.np}
  The decision version of \dsknap is NP-complete.
\end{theorem}
\begin{proof}
  We first notice that checking if a solution to \dsknap can be done
  in polynomial time, so this problem belongs in NP.

  We prove the NP-hardness by reduction to equal-sized 2-Partition:
  given $2n$ integers $a_i$, does there exist a set $J$ such that
  $\cardi{J} = n$ and $\sum_{i\in J} a_i = \frac{1}{2} \sum_i a_i$ ?
  From an instance $\mathcal{I}$ of this problem, we build the
  following instance $\mathcal{I}'$ of \dsknap: $p_i = 1+a_i$ and $s_i
  = a_i$, with $M = n$ and $C = \frac{1}{2} \sum_i a_i$. We claim that
  $\mathcal{I}$ has a solution if and only if $\mathcal{I}'$ has a
  solution of profit at least $n + C$.  Indeed, if $\mathcal{I}$ has a
  solution $J$, then $J$ is a valid solution for $\mathcal{I}'$ with
  profit $n+C$.

  Reciprocally, if $\mathcal{I}'$ has a solution $J$ with weights
  $x_i$ and profit $p \geq n+C$, then
  \begin{align*} 
    p &= \sum_{i\in J} x_ip_i =  \sum_{i \in J} x_i + \sum_{_ \in J} x_is_i\\
    & \leq \sum s_i + C \quad\text{since $J$ is a valid solution}\\
  \end{align*}

  We get $\sum_{i \in J} x_i + C \geq p \geq n + C$, hence $\sum_{i \in
    J} x_i \geq n$. Since $\cardi{J} \leq n$ and $x_i \leq 1$, this
  implies that all $x_i$ for $i \in J$ are equal to 1 and that
  $\cardi{J} = n$. Furthermore, $S = \sum_{i\in J} s_i$ verifies $S \leq
  B$ because $J$ is a valid solution for $\mathcal{I}'$, and $S \geq B$
  because $S = p - n$. Hence $J$ is thus a solution for $\mathcal{I}$.
\end{proof}

\begin{theorem}
  An optimal solution to \dsknap can be found in time $O(nCM)$ with a
  dynamic programming algorithm.
\end{theorem}
\begin{proof}
  We first assume that the items are sorted by non-increasing
  $\frac{p_i}{s_i}$ ratios. For any value $0\leq u\leq C$, $0\leq
  l\leq M$ and $0\leq i\leq n$, let us define $P(u, l, i)$ to be the
  maximum profit that can be reached with a capacity $u$, with at most
  $l$ items, and by using only items numbered from $1$ to $i$, without
  splitting. We can easily derive that
  \begin{multline*}
    P(u, l, i+1) = \\
    \begin{cases}
      0 & \text{if $l = 0$ or $i = 0$}\\
      \begin{split}
        \max(&P(u, l, i), \\
        &P(u-s_{i+1}, l-1, i) + p_{i+1})
      \end{split} & \text{if $u
        \geq s_{i+1}$ and $l > 0$}\\
      P(u, l, i) & \text{otherwise}
    \end{cases}
  \end{multline*}
  We can thus recursively compute $P(u, l, i)$ in $O(nCM)$ time. 

  Using Remark~\ref{rem.split.dominant}, we can use $P$ to compute
  $P'(i)$, defined as the maximum profit that can be reached in a
  solution where $i$ is split: 
  $$P'(i) = \max_{0 < x < 1} P(C-x s_i,
  M-1, i-1) + xp_i$$ 
  Computing $P'$ takes $O(nC)$ time. The optimal
  profit is then the maximum value between $P(C, M, n)$ (in which case
  no item is split) and $\max_{1\leq i \leq n} P'(i)$ (in this case
  item $i$ is split).
\end{proof}



\bigskip

\begin{algorithm}
\caption{\label{algo.packing}Summary of our two-step packing heuristic}
\begin{algorithmic}[1]
\Function{Homogeneous}{$\pen{i}$}
\State Binary Search for 
$\pd$ satisfying eq.~\eqref{eq.pd}
\State Compute $\nb{i} = g_i(X)$, then $\allh{i}$ according to eq.~\eqref{eq.homo.r} 
\State \Return{$\nb{i}, \allh{i}$}
\EndFunction
\Function{Heuristic}{}
\State Compute \pen{i} using $\zr{i}$ from normal law
\Repeat\State $\nb{i}, \allh{i} \gets \textsc{Homogeneous}(\pen{i})$
\State Compute $\nb{i}'$ from binomial distribution
\State Compute $\pen{i}$ from eq.~\eqref{eq.homo.r} with $\nb{i}'$ and $\allh{i}$
\Until{no $\pen{i}$ has changed by more than $\eps$}
\State $\mathcal{C} \gets$ greedy configurations
\Repeat 
\State Solve Eq(\ref{pb.colgen}) with configurations from $\mathcal{C}$
\State Get dual variables $p_i$
\State $c\gets $ solution of $\dsknap(p_i)$
\State $\mathcal{C} \gets \mathcal{C} \cup \{c\}$
\Until{Solution of \dsknap has profit $\leq 1$}
\EndFunction
\end{algorithmic}
\end{algorithm}

In this section, we have proposed a two-step algorithm to solve the
allocation problem under reliability constraints. The complete
algorithm is summarized in Algorithm~\ref{algo.packing}. The execution
time of the first loop is linear in $\ns$, and in practice it is
executed at most $10$ times, so the first step is linear (and in
practice very fast). The execution time of the second loop is also
polynomial: solving a linear program on rational numbers is very
efficient, and the dynamic program has complexity
$O(\ns\mem\cpu)$. Furthermore, in practice the number of
generated configurations is very low, of the order of \ns, whereas the
total number of possible configurations is $O(\ns^{\mem})$.

In the following (especially in Section~\ref{sec.simus}), we will
evaluate the performance of this algorithm on several
randomly-generated scenarios, in terms of running time and number of required
machines. However, since in our approach the reliability
constraints are taken into account in an approximate way, we are also
interested in evaluating the resulting reliability of generated allocations. This is done in the next Section.



\section{Reliability estimation}
\label{proba.est}

From previous results~\cite{beaumont:hal-00743524}, we know that computing
exactly the reliability of an allocation is a difficult problem: it is
actually a \#P-complete problem. A pseudo-polynomial dynamic
programming algorithm was proposed to solve this problem. However,
this algorithm assumes integral allocations and its running time is
linear in the number of machines. This makes it not feasible to use it
in the context of our paper, with large platforms and several hundreds
of services to estimate.


In this section, we explore another way to estimate the reliability
value of a given configuration. The algorithm presented here is
adapted from Algorithm 2.2 in~\cite{algo-rare}. The objective is to
compute a good approximation of the probability that a given service
fails, \ie $\pfail{i} = \pro{\alcp{i} < \demand{i}}$. A
straightforward approach for this kind of estimation is to generate a
large sample of scenarios for \alcp{i}, and compute the proportion of
scenarios in which $\alcp{i} < \demand{i}$. However, this strategy fails if the events that we aim at detecting are very rare (as reliability requirements violations are in our context since \pfail{i} is expected to be of the order of
\rel{i}, which could be or order $10^{-6}$ or lower). Indeed, it would
require to generate a very large number of samples (more than $10^{8}$
for the estimate of $\rel{i}=10^{-6}$). A more sensible approach, as
described in~\cite{algo-rare}, is to decompose the computation of
\pfail{i} into a product of conditional probabilities, whose values
are reasonably not too small (typically around $10^{-1}$), and hence can
be estimated with smaller sampling sizes. With this idea, estimating a
reliability of $10^{-6}$ would require $6$ iterations, each of which
uses a sampling of size $10^3$, which dramatically reduces the $10^8$
sampling size required by the direct approach.


\subsection{Formal description}

Since computing the reliability of each service can be done
independently, we consider here a given service \serv{i}, and for the
ease of notations, we omit the index $i$ until the end of this
section.

We rewrite \pfail{} in the following way:
\begin{multline*} \pfail{} = \pro{\alcp{} < \demt{}{1}} \\\times
\prod_{k \in \inter[2]{\nt{}}} \pro{\alcp{} < \demt{}{k} | \alcp{} <
\demt{}{k-1}},
\end{multline*} where the \demt{}{k}'s are thresholds such that
\[ \demt{}{1} > \demt{}{2} > \dots > \demt{}{\nt{}} = \demand{}. \] We
will see in the next subsection that those thresholds can be chosen
on-the-fly.

In order to express the probability distribution of \alcp{}, we use
the configuration description as in the previous section. We recall
that \allt{}{c} is the CPU allocated to service \serv{} in
configuration \conf{c} and \lam{c} is the number of machines that
follows this configuration. Then we have
\[ \alcp{} \sim \sum_{c | \serv{} \in \conf{c}} \allt{}{c} \times
\bino{\lam{c}}{1-\fail}.\] 

After a straightforward renumbering, and by denoting \nc{} the number
of different configurations in which the service \serv{} appears, we
obtain $\alcp{} = \sum_{c=1}^{\nc{}} \allt{}{c} \rvbino{c}$, such that
for all $c \in \inter{\nc{}}$, $\rvbino{c} \sim
\bino{\lam{c}}{1-\fail}$.

A value of the random variable \alcp{} is thus fully described by the
value of the random vector variable $(\rvbino{1},
\dots,\rvbino{\nc{}})$. In addition, when $Y=(\valbino{1}{}, \dots,
\valbino{\nc{}}{})$ is a value of this random vector, we define
$\sumall{Y} = \sum_c \valbino{c}{}$.

\subsection{Full algorithm}

The idea of the algorithm (described in details in
Algorithm~\ref{algo.all}) is to maintain a sample of random vectors
$Y_1, \dots, Y_\sasi$, distributed according to the original
distribution, conditional to $\sumall{Y} < \demt{}{k}$ at each step
$k$. Obtaining the sample for step $k+1$ is done in three
steps. First, the value of $\demt{}{k+1}$ is computed so that a $10\%$
fraction of the current sample satisfy $\sumall{Y} < \demt{}{k+1}$
(line~\ref{algo.all.thres}). Then, in the \textsc{Bootstrap} step, we
keep only the values that satisfies $\sumall{Y_s} < \demt{}{k+1}$, and
draw uniformly at random $\sasi$ vectors from this set, with
replacement. Finally, in the \textsc{Resample} step, we modify each
vector $Y_s$ of this set, one coordinate after the other, by
generating a new value $\valbino{c}{s}$ according to the distribution
$\bino{\lam{c}}{1-\fail}$ conditional to $\sumall{Y_s} <
\demt{}{k+1}$. This ensures that the new sample is distributed
according to the required conditional distribution. At each step, an
unbiased estimate of \pro{\alcp{} < \demt{}{k+1} | \alcp{} <
  \demt{}{k}} is the number of vectors which satisfy $\sumall{Y_s} <
\demt{}{k+1}$, divided by the total sampling size \sasi.

The \textsc{Resample} step is described in more details in
Algorithm~\ref{algo.resample}. In order to generate a new value
$\valbino{c}{s}$ conditional to $\sumall{Y_s} < \demt{}{k+1}$, it is
sufficient to compute the total CPU allocated in the other
configurations $\dcur= \sum_{c' \neq c} \allt{}{c'}$: then, the
condition is equivalent to $\valbino{c}{s} \leq
\frac{\demt{}{k+1}-\dcur}{\allt{}{c}}$, and this amounts to generating
according to a truncated binomial distribution.


\begin{algorithm}[h]
  \caption{\label{algo.resample}Resampling}
  \begin{algorithmic}[1] \Function{Resample}{$\thres$} \For{$s \in
\inter{\sasi}$} \For{$c \in \inter{\nc{}}$} \State $\dcur \gets
\sum_{c' \neq c} \allt{}{c'} \valbino{c'}{s}$ \State $\minmach \gets
\frac{\thres-\dcur}{\allt{}{c}}$ \State \label{algo.l.conbino}Draw
\valbino{c}{s} following \bino{\lam{c}}{1-\fail} conditional to
$\valbino{c}{s}< \minmach$
    \EndFor
    \EndFor
    \EndFunction
  \end{algorithmic}
\end{algorithm}

\begin{algorithm}[h]
  \caption{\label{algo.all}Adaptive Reliability estimate}
  \begin{algorithmic}[1] \Function{ReliabilityEstimate}{} 
\State$\currel \gets 1$
 \State Draw a sample $(Y_1,\dots,Y_\sasi)$ 
\State
$\curthr \gets \argmin_t \{ \frac{\# \{ Y_s | \sumall{Y_s} < t
\}}{\sasi} > 10\% \} $ 
\State $\curthr \gets \max(\curthr, \demand{})$
\State $\currel \gets \currel \times \frac{\# \{ Y_s |
  \sumall{Y_s} < \curthr \}}{\sasi}$ 
\While{$\curthr > \demand{}$}
\State \textsc{Bootstrap}$(\curthr)$ 
\State \textsc{Resample}$(\curthr)$ 
\label{algo.all.thres}\State $\curthr \gets \argmin_t \{\frac{\# \{ Y_s | \sumall{Y_s} < t
  \}}{\sasi} > 10\% \} $ 
\State $\curthr \gets \max(\curthr, \demand{})$
 \State $\currel \gets \currel \times \frac{\# \{ Y_s | \sumall{Y_s} < \curthr \}}{\sasi}$
 \EndWhile 
\State \Return $\currel$
\EndFunction
  \end{algorithmic}
\end{algorithm}

\section{Simulations}
\label{sec.simus}

In this section, we perform a large set of simulations with two main
objectives.  On the one hand, we assess the performance of Algorithm~\ref{algo.packing} and we observe the number of
machines used, the execution time and the compliance with the
reliability requests. On the other hand, we study the behavior of the
reliability estimation algorithm in Section~\ref{sec.simus.estimate}.

Simulations were conducted using a node based on two quad-core Nehalem
Intel Xeon X5550, and the source code of all heuristics and simulations
is publicly available on the Web~\cite{simus-liopaul}.

\subsection{Resource Allocation Algorithms}
\label{sec.nosharing}
\label{sec.simus.alloc}

\newcommand{\hdu}{{\small\textsf{no\_sharing}}\xspace}
\newcommand{\hcgpd}{{\small\textsf{colgen\_pd}}\xspace}
\newcommand{\hcg}{{\small\textsf{colgen}}\xspace}
\newcommand{\hcgf}{{\small\textsf{colgen\_float}}\xspace}

\label{sec.nosharing}
\label{sec.simus.alloc}

In order to give a point of comparison to describe the contribution of
our algorithm, we have designed an additional simple greedy heuristic.
This heuristic is based on an exclusivity principle: two different
services are not allowed to share the same machine. Each service is
thus allocated the whole CPU power of some number of machines. The
appropriate number of machines for a given service is the minimum
number of machines that have to be dedicated to this service, so that
the reliability constraint is met. This can be easily computed using
the cumulative distribution function of a binomial
distribution~\cite{lib-binomial} and a binary search. We greedily
assign the necessary number of machines to each service and obtain an
allocation that fulfills all reliability constraints. This heuristic is
named \hdu.

In the algorithms based on column generation, the linear program for Eq(\ref{pb.colgen}) is solved in rational numbers, because its
integer version is too costly to solve optimally. An integer solution
is obtained by rounding up all values of a given
configuration, so as to ensure that the reliability constraints are
still fulfilled. This increases the number of used machines, so we
also keep the rational solution as a lower bound. The ceiled variant
that uses a linear program to solve the \dsknap problem is called
\hcg, while \hcgpd runs the dynamic programming algorithm. Finally,
\hcgf denotes the lower bound. The legend that applies for all the
graphs in this Section is given in Figure~\ref{fig.key}.
\begin{figure}[h!]
  \centering
  \includegraphics[width=.45\textwidth]{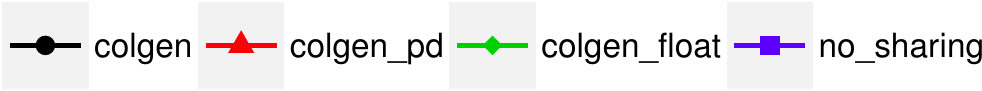}
  \caption{Key}
  \label{fig.key}
\end{figure}

\newcommand{\scunifo}{\textsc{Uniform}\xspace}
\newcommand{\scbival}{\textsc{Bivalued}\xspace}

\subsection{Simulation settings}
\label{sec.simus.set}

We explore two different kinds of scenarios in our experiments. In the
first scenario, called \scunifo, we envision a cloud where all services have close demands. The CPU demand of a service is
drawn uniformly between the equivalent CPU capacity of 5 and 50 machines, and we
vary the number of services from 20 to 300. In this last case, the overall
required number of machines is more than 8000 on average. In the
\scbival scenario, two classes of services request for resources. Three
big services set their demand between 900 and 1100 machines, while the
298 other clients need from 5 to 15 machines, so that the machines are
fairly shared between the two classes of services. This leads to more
than 6000 machines in average.

In both cases, the reliability request for each service is taken to be
$10^{-X}$, where $X$ is drawn uniformly between $2$ and $8$. On the
platform side, we vary the memory capacity of the machines from 5 to
10, and the failure probability of a machine is set to $0.01$.

\subsection{Number of required machines}
\label{sec.simus.nm}

The number of required machines by each heuristic is depicted in
Figure~\ref{fig.nm}.
The first observation is that the rounding step increases the number
of machines used by at most $2.5\%$. This can be explained by noting
that the number of different configurations that are actually used in
a solution of \hcg is less than the number of services, hence each
configuration is used a relatively large number of times.

Another observation is that the \hdu heuristic is sensitive to the
memory capacity, and, as expected, its quality decreases compared to
the column generation heuristics (by construction, given a set of
services, \hdu will return the same solution whatever the memory
capacity) and becomes $12.5\%$ worse than the lower bound in the
\scunifo scenario. Services are indeed distributed on more machines
in the column generation heuristics, and hence need less replication
to fulfill the reliability constraints.

Finally, concerning the final number of machines, \hcg and \hcgpd are very
close, which shows that the discretization required for the dynamic
programming algorithm does not induce a noticeable loss of quality. 

\begin{figure}[t!]
  \centering
\subfiguretopcaptrue
\subfigure[\scbival scenario]{%
  \includegraphics[width=.80\textwidth]{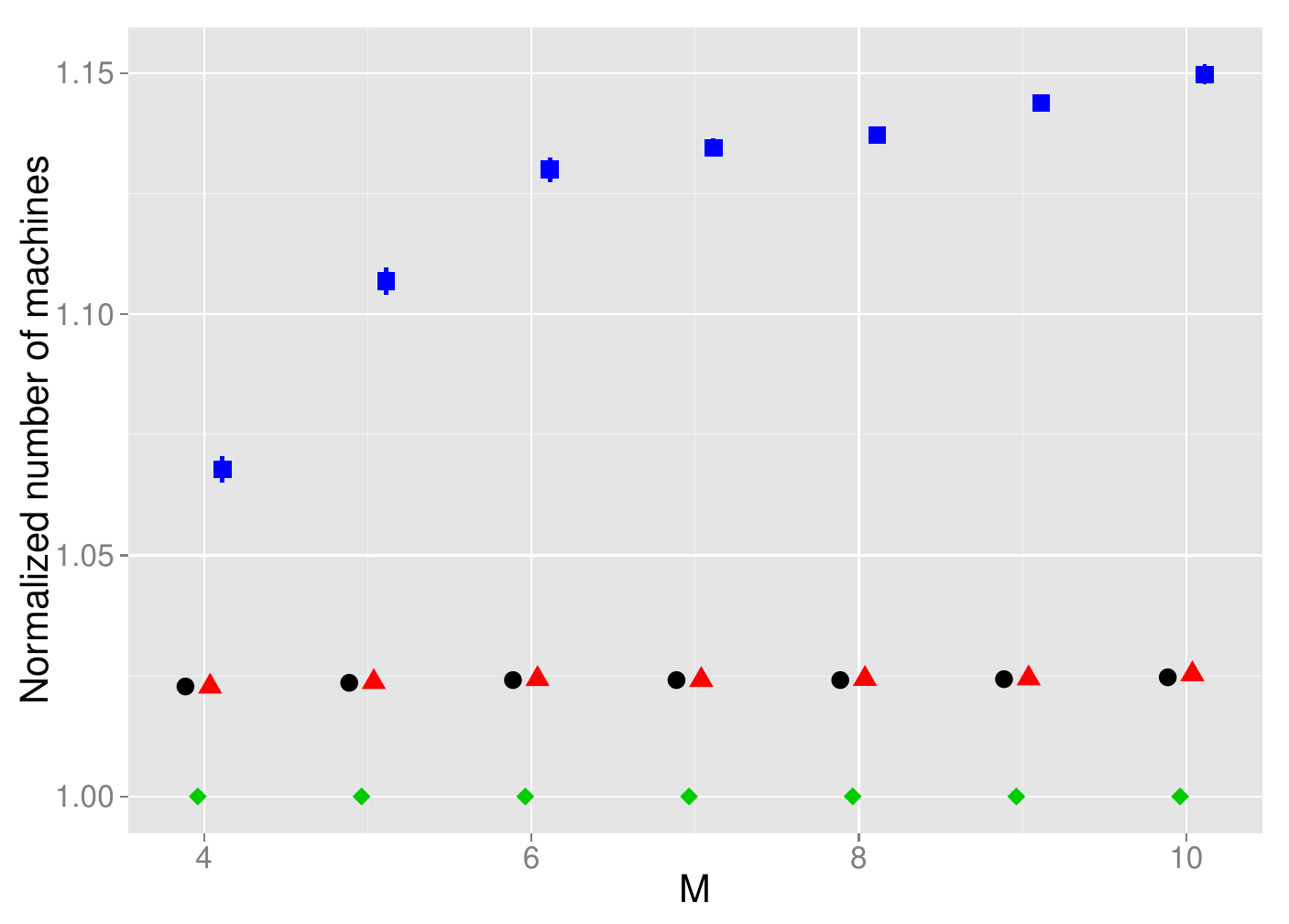}
}
\subfiguretopcapfalse%
\subfigure[\scunifo scenario]{%
  \includegraphics[width=.80\textwidth]{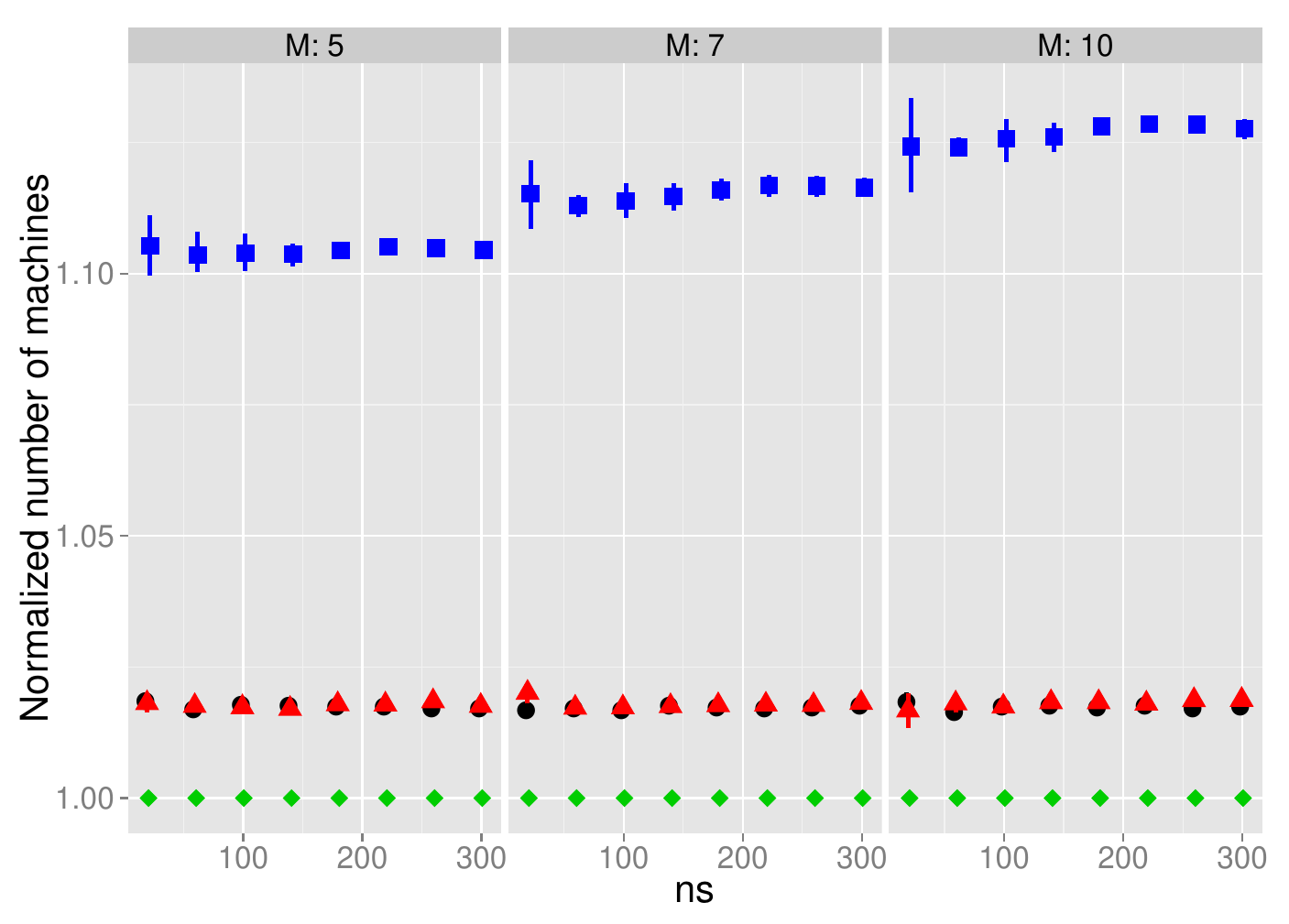}
}
\caption{Number of required machines}
\label{fig.nm}
\end{figure}


\newcommand{\se}[1]{\ema{#1\,s}}
\newcommand{\mi}[1]{\ema{#1\,\mathit{min}}}
\subsubsection{Execution time}
\label{sec.simus.time}

We represent in Figure~\ref{fig.time}
the execution time of all heuristics, in both scenarios.  There
appears a difference between \hcg and \hcgpd. In the \scunifo scenario
with low memory capacity, the execution time of \hcg is at the same
time very high and very unstable (it may reach \mi{20}), while \hcgpd
remains under \se{30}. With higher memory capacity, the execution time
of \hcg improves while \hcgpd gets slower; they become similar when 10
services are allowed on the same machine.

On the other hand, in the \scbival case, the execution time of both
variants increase when the memory capacity increases, and \hcgpd is
always better than \hcg.

Finally in all cases \hdu confirms that it is a very cheap heuristic,
and its execution time never exceeds \se{0.05}.

\begin{figure}[t!]
  \centering
\subfiguretopcaptrue
\subfigure[\scbival scenario]{%
\includegraphics[width=.80\textwidth]{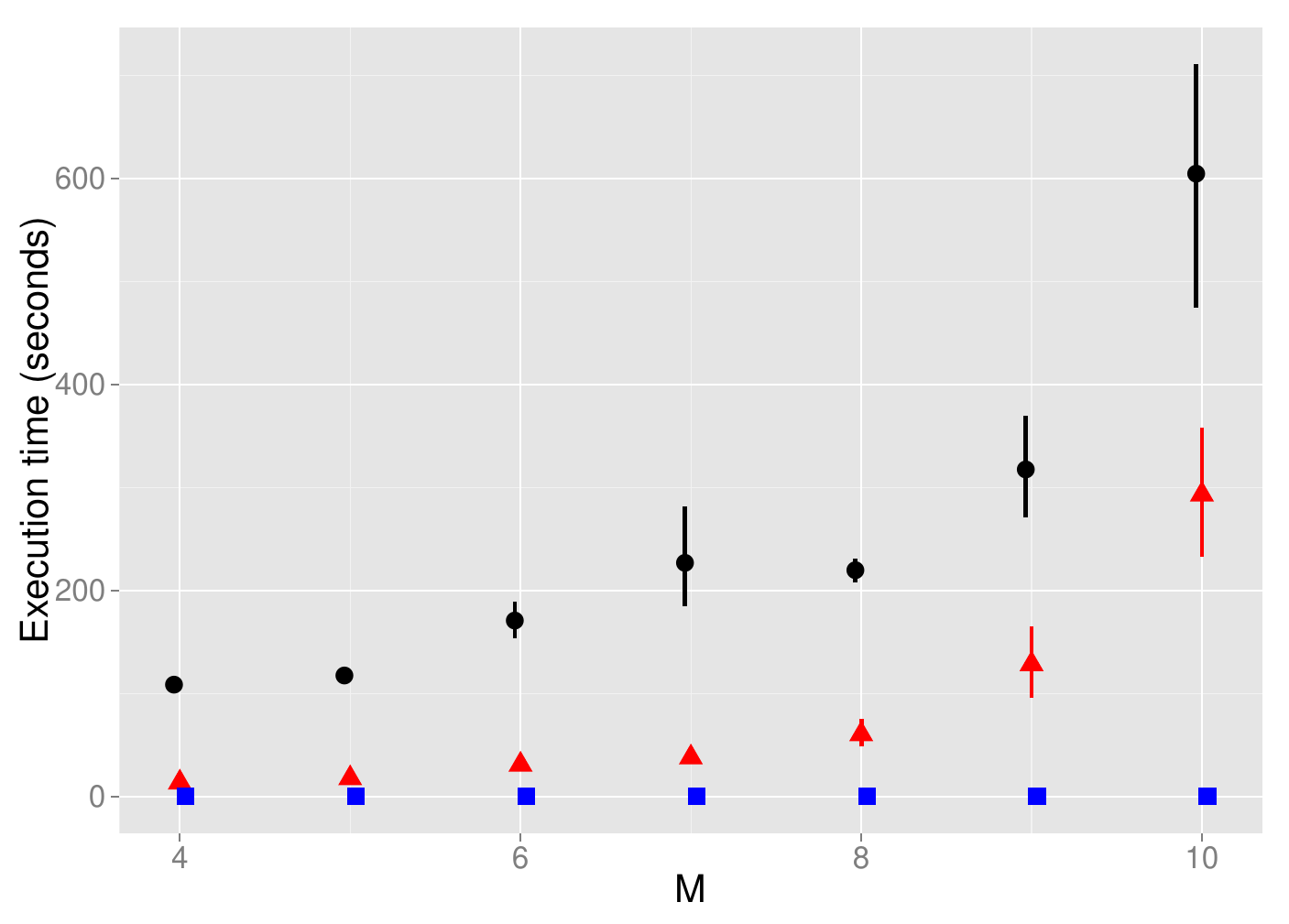}
}
\subfiguretopcapfalse%
\subfigure[\scunifo scenario]{%
\includegraphics[width=.80\textwidth]{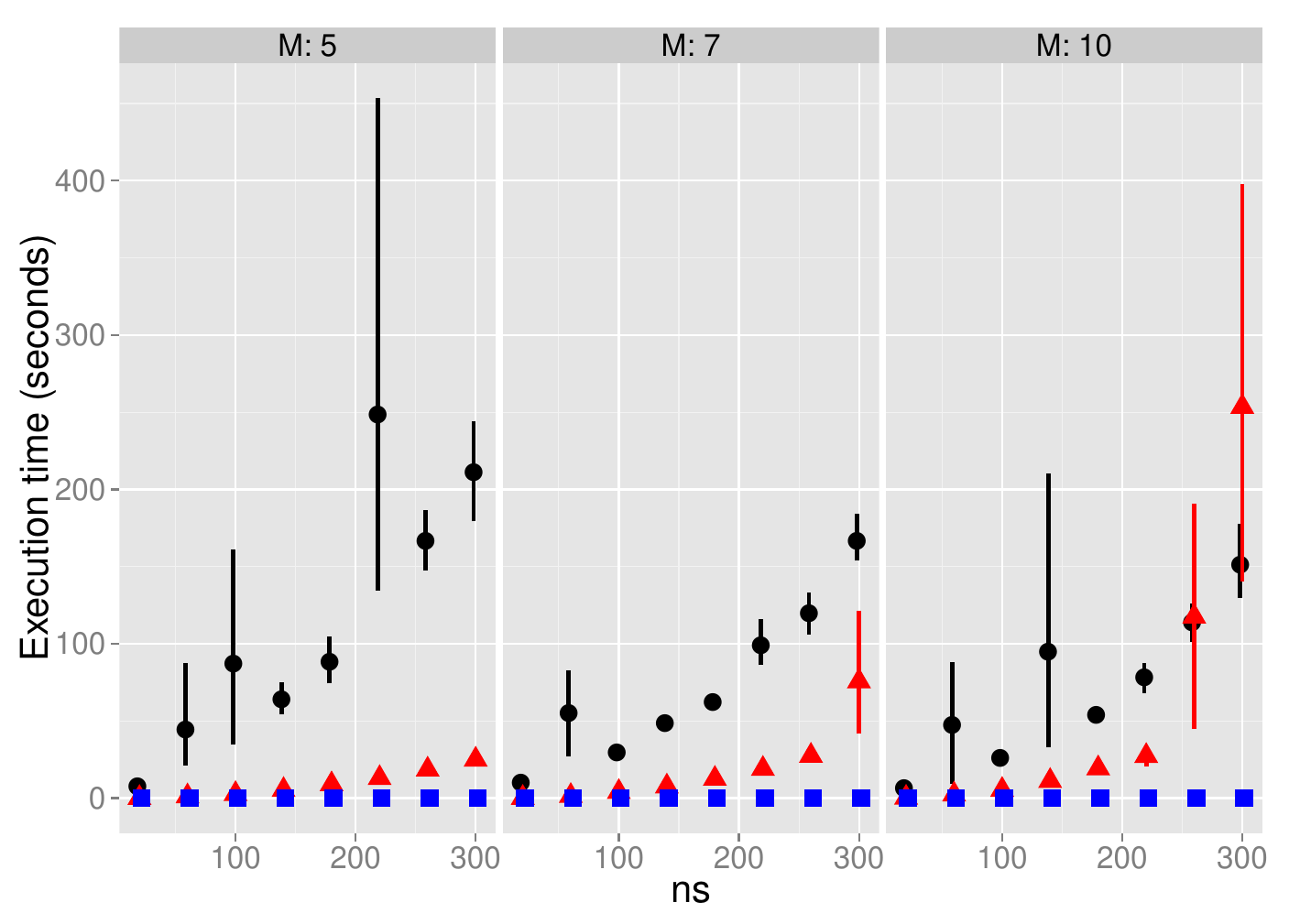}
}
\caption{Execution time}
\label{fig.time}
\end{figure}


\subsection{Reliability}
\label{sec.simus.rel}

We represent in Figure~\ref{fig.rel}
the compliance of the services with their reliability constraint. This
constraint is met if the point is below the black straight
line. Because of the rounding step, the column generation heuristics
fulfill easily the reliability bounds in average. Moreover, we observed (results are not displayed here due to lack of space) that the worse cases are very close to the line but remain below it,
which is intended by construction of the heuristics.

In the \scbival scenario, the column generation heuristics produce
allocations that are even more reliable than in the \scunifo
scenario. This come from the fact that the big services are assigned
to a large number of different configurations, and hence gain even
more CPU after the rounding step.

There is a clear double advantage to the column generation-based algorithm,
and especially \hcgpd, against \hdu: it leads to more reliable
allocations with a noticeably smaller number of machines.

\begin{figure}[t!]
  \centering
\subfiguretopcaptrue
\subfigure[\scbival scenario]{%
 \includegraphics[width=.80\textwidth]{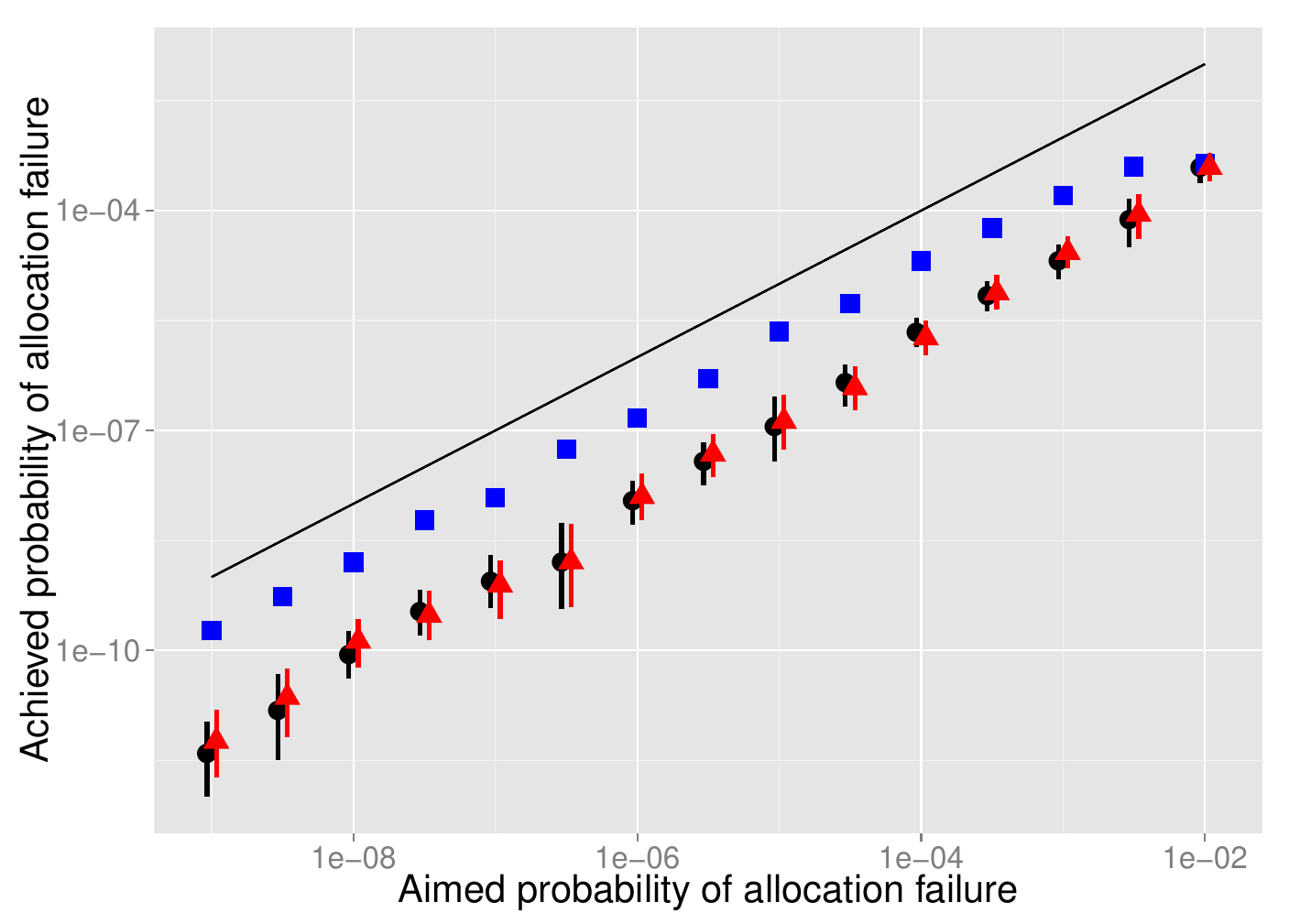}
}
\subfiguretopcapfalse%
\subfigure[\scunifo scenario]{%
\includegraphics[width=.80\textwidth]{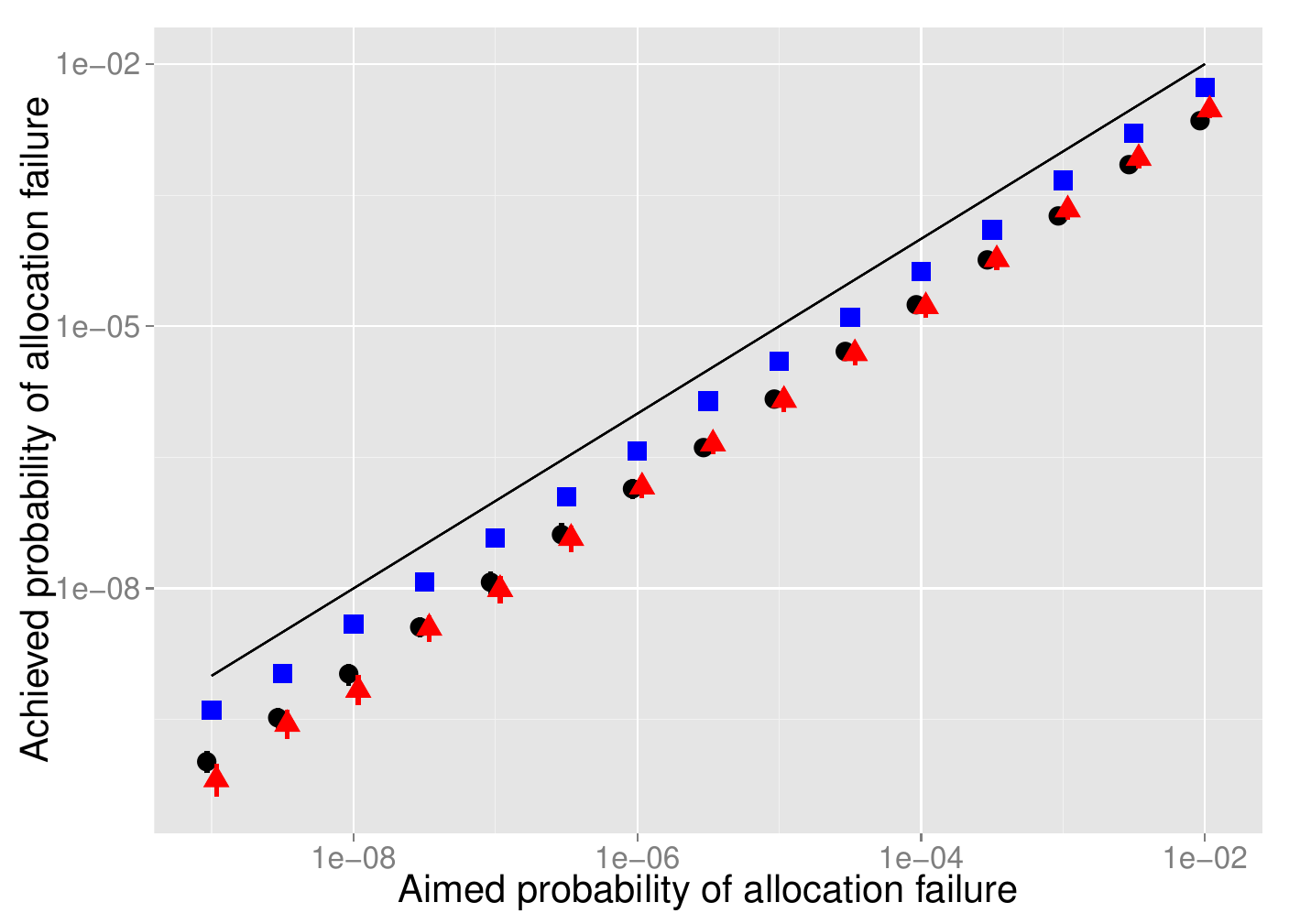}
}
\caption{Reliability constraints}
\label{fig.rel}
\end{figure}


\subsection{Probability estimation algorithm}
\label{sec.simus.estimate}

In Figure~\ref{fig.at} we plot the execution time of the probability
estimation algorithm of the reliability of a service as a function of the actual
failure probability of the allocation. Obviously, when the failure
probability decreases, the execution time increases, since the event
that we try to capture is rarer. This estimate algorithm is efficient,
since it can estimate an event of probability $10^{-17}$ in about
\se{40}. In our particular case, it is possible to further lower
this execution time if we focus only on checking whether the failure
probability requirement is not exceeded. Indeed, the requirements are
lower than $10^{-9}$. By stopping the algorithm early, it is possible
to determine whether the allocation is valid in at most \se{15}.

We can remark that the estimate of failure probabilities in a solution
returned by \hdu is not expensive, since each service is allocated to
exactly one configuration. Hence, at each step of the algorithm, we have
only one draw for one binomial distribution, which is not the case with
solutions that are provided by the other heuristics.

\begin{figure}[t]
  \centering
  \includegraphics[width=.80\textwidth]{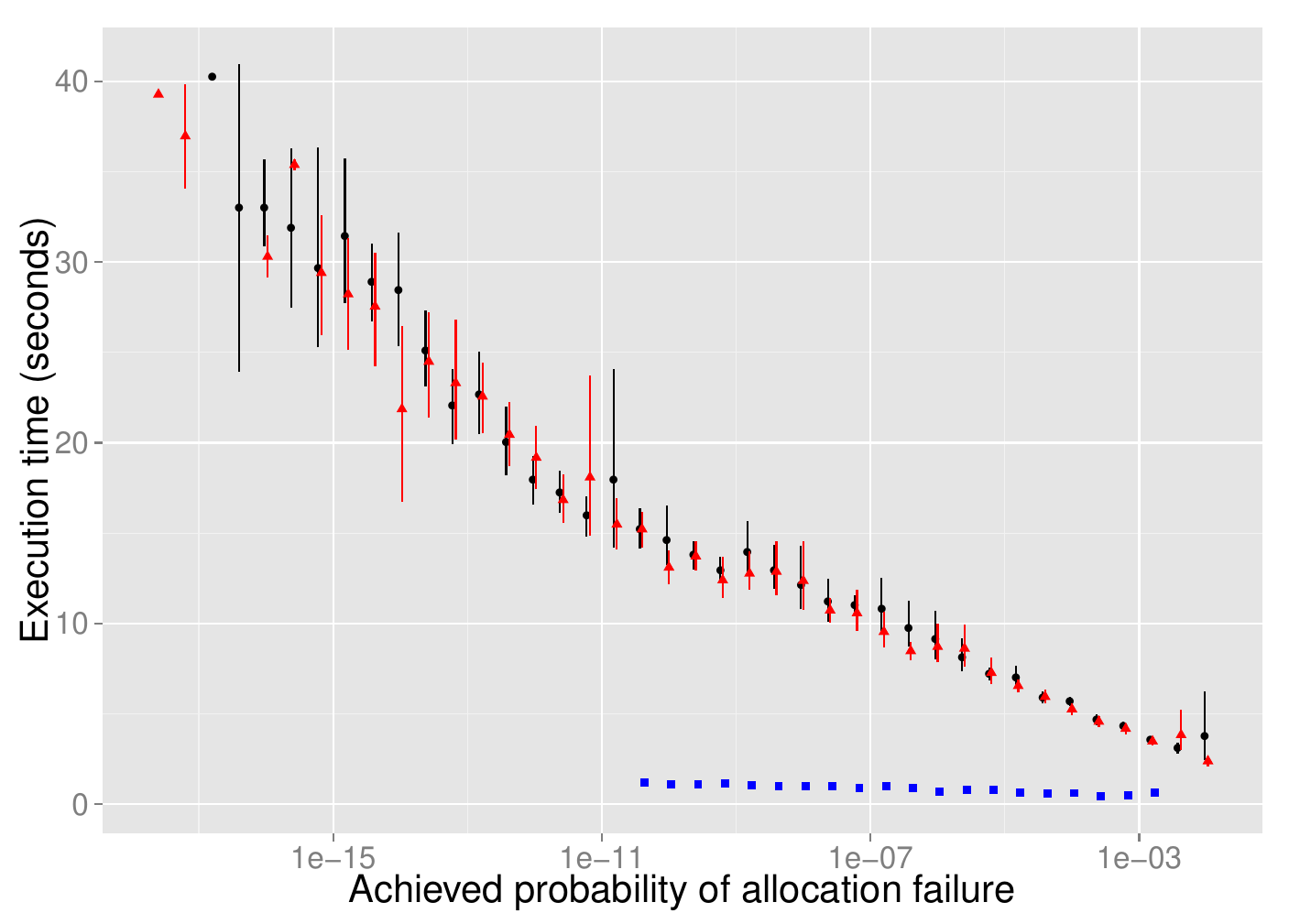}
  \caption{Execution Time of Probability Estimation Algorithm}
  \label{fig.at}
\end{figure}

\section{Conclusion}
\label{sec.conclusion} 

The evolution of large computing platforms makes fault-tolerance issues crucial. With this respect, this paper considers a simple setting, with a set of services handling requests on an homogeneous cloud platform. To deal with fault tolerance issues, we assume that each service comes with a global demand and a reliability constraint. Our contribution follows two directions. First, with borrow and adapt from Applied Probability literature sophisticated  techniques for estimating the failure probability of an allocation, that remains efficient even if the considered probability is very low ($10^{-10}$ for instance). Second, we borrow and adapt from the Mathematical Programming and Operations Research literature the use of Column Generation techniques, that enable to solve efficiently some classes of linear programs. The use of both techniques enables to solve in an efficient manner the resource allocation problem that we consider, under a realistic settings (both in terms of size of the problem and characteristics of the applications, for instance discrete unsplittable memory constraints) and we believe that it can be extended to many other fault-tolerant allocation problems.

\bibliographystyle{IEEEtran} \bibliography{IEEEabrv,biblio}

\begin{thebibliography}{10}
\providecommand{\url}[1]{#1}
\csname url@samestyle\endcsname
\providecommand{\newblock}{\relax}
\providecommand{\bibinfo}[2]{#2}
\providecommand{\BIBentrySTDinterwordspacing}{\spaceskip=0pt\relax}
\providecommand{\BIBentryALTinterwordstretchfactor}{4}
\providecommand{\BIBentryALTinterwordspacing}{\spaceskip=\fontdimen2\font plus
\BIBentryALTinterwordstretchfactor\fontdimen3\font minus
  \fontdimen4\font\relax}
\providecommand{\BIBforeignlanguage}[2]{{%
\expandafter\ifx\csname l@#1\endcsname\relax
\typeout{** WARNING: IEEEtran.bst: No hyphenation pattern has been}%
\typeout{** loaded for the language `#1'. Using the pattern for}%
\typeout{** the default language instead.}%
\else
\language=\csname l@#1\endcsname
\fi
#2}}
\providecommand{\BIBdecl}{\relax}
\BIBdecl

\bibitem{shih2010performance}
W.~Shih, S.~Tseng, and C.~Yang, ``Performance study of parallel programming on
  cloud computing environments using mapreduce,'' in \emph{International
  Conference on Information Science and Applications (ICISA)}.\hskip 1em plus
  0.5em minus 0.4em\relax IEEE, 2010, pp. 1--8.

\bibitem{dean2008mapreduce}
J.~Dean and S.~Ghemawat, ``Mapreduce: Simplified data processing on large
  clusters,'' \emph{Communications of the ACM}, vol.~51, no.~1, pp. 107--113,
  2008.

\bibitem{zaharia2008improving}
M.~Zaharia, A.~Konwinski, A.~Joseph, R.~Katz, and I.~Stoica, ``Improving
  mapreduce performance in heterogeneous environments,'' in \emph{Proceedings
  of the 8th USENIX conference on Operating systems design and
  implementation}.\hskip 1em plus 0.5em minus 0.4em\relax USENIX Association,
  2008, pp. 29--42.

\bibitem{bougeret2011checkpointing}
M.~Bougeret, H.~Casanova, M.~Rabie, Y.~Robert, and F.~Vivien, ``Checkpointing
  strategies for parallel jobs,'' in \emph{High Performance Computing,
  Networking, Storage and Analysis (SC), 2011 International Conference
  for}.\hskip 1em plus 0.5em minus 0.4em\relax IEEE, 2011, pp. 1--11.

\bibitem{cappello2010checkpointing}
F.~Cappello, H.~Casanova, and Y.~Robert, ``Checkpointing vs. migration for
  post-petascale supercomputers,'' \emph{ICPP'2010}, 2010.

\bibitem{bouteiller2013multi}
A.~Bouteiller, F.~Cappello, J.~Dongarra, A.~Guermouche, T.~H{\'e}rault, and
  Y.~Robert, ``Multi-criteria checkpointing strategies: response-time versus
  resource utilization,'' in \emph{Euro-Par 2013 Parallel Processing}.\hskip
  1em plus 0.5em minus 0.4em\relax Springer, 2013, pp. 420--431.

\bibitem{wang2009improving}
C.~Wang, Z.~Zhang, X.~Ma, S.~S. Vazhkudai, and F.~Mueller, ``Improving the
  availability of supercomputer job input data using temporal replication,''
  \emph{Computer Science-Research and Development}, vol.~23, no. 3-4, pp.
  149--157, 2009.

\bibitem{ferreira2011evaluating}
K.~Ferreira, J.~Stearley, J.~Laros~III, R.~Oldfield, K.~Pedretti,
  R.~Brightwell, R.~Riesen, P.~Bridges, and D.~Arnold, ``Evaluating the
  viability of process replication reliability for exascale systems,'' in
  \emph{Proceedings of 2011 International Conference for High Performance
  Computing, Networking, Storage and Analysis}.\hskip 1em plus 0.5em minus
  0.4em\relax ACM, 2011, p.~44.

\bibitem{zhang2010cloud}
Q.~Zhang, L.~Cheng, and R.~Boutaba, ``{Cloud computing: state-of-the-art and
  research challenges},'' \emph{Journal of Internet Services and Applications},
  vol.~1, no.~1, pp. 7--18, 2010.

\bibitem{armbrust2009above}
M.~Armbrust, A.~Fox, R.~Griffith, A.~Joseph, R.~Katz, A.~Konwinski, G.~Lee,
  D.~Patterson, A.~Rabkin, I.~Stoica \emph{et~al.}, ``{Above the clouds: A
  berkeley view of cloud computing},'' \emph{EECS Department, University of
  California, Berkeley, Tech. Rep. UCB/EECS-2009-28}, 2009.

\bibitem{CirneInvited}
W.~Cirne and E.~Frachtenberg, ``Web-scale job scheduling,'' in \emph{Job
  Scheduling Strategies for Parallel Processing}.\hskip 1em plus 0.5em minus
  0.4em\relax Springer, 2013, pp. 1--15.

\bibitem{amazon}
``Amazon elastic compute cloud (amazon ec2),''
  \url{http://aws.amazon.com/fr/ec2/}.

\bibitem{van2009sla}
H.~Van, F.~Tran, and J.~Menaud, ``{SLA-aware virtual resource management for
  cloud infrastructures},'' in \emph{IEEE Ninth International Conference on
  Computer and Information Technology}.\hskip 1em plus 0.5em minus 0.4em\relax
  IEEE, 2009, pp. 357--362.

\bibitem{calheiros2009heuristic}
R.~Calheiros, R.~Buyya, and C.~De~Rose, ``{A heuristic for mapping virtual
  machines and links in emulation testbeds},'' in \emph{ICPP}.\hskip 1em plus
  0.5em minus 0.4em\relax IEEE, 2009, pp. 518--525.

\bibitem{christopherTPDS}
O.~Beaumont, L.~Eyraud-Dubois, H.~Rejeb, and C.~Thraves, ``{Heterogeneous
  Resource Allocation under Degree Constraints},'' \emph{IEEE Transactions on
  Parallel and Distributed Systems}, 2012.

\bibitem{berl2010energy}
A.~Berl, E.~Gelenbe, M.~Di~Girolamo, G.~Giuliani, H.~De~Meer, M.~Dang, and
  K.~Pentikousis, ``{Energy-efficient cloud computing},'' \emph{The Computer
  Journal}, vol.~53, no.~7, p. 1045, 2010.

\bibitem{beloglazov2010energy}
A.~Beloglazov and R.~Buyya, ``{Energy efficient allocation of virtual machines
  in cloud data centers},'' in \emph{2010 10th IEEE/ACM International
  Conference on Cluster, Cloud and Grid Computing}.\hskip 1em plus 0.5em minus
  0.4em\relax IEEE, 2010, pp. 577--578.

\bibitem{GareyJohnson}
M.~R. Garey and D.~S. Johnson, \emph{Computers and Intractability, a Guide to
  the Theory of {NP}-Completeness}.\hskip 1em plus 0.5em minus 0.4em\relax W.\
  H. Freeman and Company, 1979.

\bibitem{EpsteinvanStee07binpackresaumg}
L.~Epstein and R.~van Stee, ``Online bin packing with resource augmentation.''
  \emph{Discrete Optimization}, vol.~4, no. 3-4, pp. 322--333, 2007.

\bibitem{Hochbaum97}
D.~Hochbaum, \emph{Approximation Algorithms for NP-hard Problems}.\hskip 1em
  plus 0.5em minus 0.4em\relax PWS Publishing Company, 1997.

\bibitem{barnhart1998branch}
C.~Barnhart, E.~L. Johnson, G.~L. Nemhauser, M.~W. Savelsbergh, and P.~H.
  Vance, ``Branch-and-price: Column generation for solving huge integer
  programs,'' \emph{Operations research}, vol.~46, no.~3, pp. 316--329, 1998.

\bibitem{desrosiers2005primer}
J.~Desrosiers and M.~E. L{\"u}bbecke, \emph{A primer in column
  generation}.\hskip 1em plus 0.5em minus 0.4em\relax Springer, 2005.

\bibitem{replication_availability}
K.~Ranganathan, A.~Iamnitchi, and I.~Foster, ``Improving data availability
  through dynamic model-driven replication in large peer-to-peer communities,''
  in \emph{Cluster Computing and the Grid, 2002. IEEE/ACM International
  Symposium on}, 2002.

\bibitem{replication_cirne_1}
D.~da~Silva, W.~Cirne, and F.~Brasileiro, ``Trading cycles for information:
  Using replication to schedule bag-of-tasks applications on computational
  grids,'' in \emph{Euro-Par 2003 Parallel Processing}, ser. Lecture Notes in
  Computer Science, H.~Kosch, L.~B\"osz\"orm\'enyi, and H.~Hellwagner,
  Eds.\hskip 1em plus 0.5em minus 0.4em\relax Springer Berlin / Heidelberg,
  2003, vol. 2790, pp. 169--180.

\bibitem{replication_datagrid}
\BIBentryALTinterwordspacing
M.~Lei, S.~V. Vrbsky, and X.~Hong, ``An on-line replication strategy to
  increase availability in data grids,'' \emph{Future Generation Computer
  Systems}, vol.~24, no.~2, pp. 85 -- 98, 2008. [Online]. Available:
  \url{http://www.sciencedirect.com/science/article/pii/S0167739X07000830}
\BIBentrySTDinterwordspacing

\bibitem{replication_DB}
\BIBentryALTinterwordspacing
H.-I. Hsiao and D.~J. Dewitt, ``A performance study of three high availability
  data replication strategies,'' \emph{Distributed and Parallel Databases},
  vol.~1, pp. 53--79, 1993, 10.1007/BF01277520. [Online]. Available:
  \url{http://dx.doi.org/10.1007/BF01277520}
\BIBentrySTDinterwordspacing

\bibitem{replication_cirne_2}
E.~Santos-Neto, W.~Cirne, F.~Brasileiro, and A.~Lima, ``Exploiting replication
  and data reuse to efficiently schedule data-intensive applications on
  grids,'' in \emph{Job Scheduling Strategies for Parallel Processing}, ser.
  Lecture Notes in Computer Science, D.~Feitelson, L.~Rudolph, and
  U.~Schwiegelshohn, Eds.\hskip 1em plus 0.5em minus 0.4em\relax Springer
  Berlin / Heidelberg, 2005, vol. 3277, pp. 54--103.

\bibitem{dongarra2009international}
J.~Dongarra, P.~Beckman, P.~Aerts, F.~Cappello, T.~Lippert, S.~Matsuoka,
  P.~Messina, T.~Moore, R.~Stevens, A.~Trefethen \emph{et~al.}, ``The
  international exascale software project: a call to cooperative action by the
  global high-performance community,'' \emph{International Journal of High
  Performance Computing Applications}, vol.~23, no.~4, pp. 309--322, 2009.

\bibitem{eesi}
``Eesi, "the european exascale software initiative", 2011,''
  \url{http://www.eesi-project.eu/pages/menu/homepage.php}.

\bibitem{cappello2009fault}
F.~Cappello, ``Fault tolerance in petascale/exascale systems: Current
  knowledge, challenges and research opportunities,'' \emph{International
  Journal of High Performance Computing Applications}, vol.~23, no.~3, pp.
  212--226, 2009.

\bibitem{beaumont:hal-00743524}
O.~Beaumont, L.~Eyraud-Dubois, and H.~Larchev{\^e}que, ``Reliable service
  allocation in clouds,'' in \emph{IPDPS'13 IEEE International Parallel \&
  Distributed Processing Symposium}, 2013.

\bibitem{valiant1979complexity}
L.~Valiant, ``The complexity of enumeration and reliability problems,''
  \emph{SIAM J. Comput.}, vol.~8, no.~3, pp. 410--421, 1979.

\bibitem{provan1983complexity}
J.~Provan and M.~Ball, ``The complexity of counting cuts and of computing the
  probability that a graph is connected,'' \emph{SIAM Journal on Computing},
  vol.~12, p. 777, 1983.

\bibitem{bodlaender2004note}
H.~Bodlaender and T.~Wolle, ``A note on the complexity of network reliability
  problems,'' \emph{UU-CS}, no. 2004-001, 2004.

\bibitem{algo-rare}
Z.~I. Botev and D.~P. Kroese, ``An efficient algorithm for rare-event
  probability estimation, combinatorial optimization, and counting,''
  \emph{Methodology and Computing in Applied Probability}, vol.~10, no.~4, pp.
  471--505, 2008.

\bibitem{beaumont2013hipc}
O.~Beaumont, P.~Duchon, and P.~Renaud-Goud, ``Approximation algorithms for
  energy minimization in cloud service allocation under reliability
  constraints,'' in \emph{HIPC'2013, IEEE international conference on High
  Performance Computing, Bangalore}, 2013.

\bibitem{chernoff}
H.~Chernoff, ``A measure of asymptotic efficiency for tests of a hypothesis
  based on the sum of observations,'' \emph{The Annals of Mathematical
  Statistics}, vol.~23, no.~4, pp. 493--507, 1952.

\bibitem{hoeffding}
W.~Hoeffding, ``Probability inequalities for sums of bounded random
  variables,'' \emph{Journal of the American Statistical Association}, vol.~58,
  no. 301, pp. 13--30, 1963.

\bibitem{nous-resilience}
O.~Beaumont, L.~Eyraud-Dubois, P.~Pesneau, and P.~Renaud-Goud, ``Reliable
  service allocation in clouds with memory and capacity constraints,'' in
  \emph{Resilience -- EuroPar workshop}, 2013.

\bibitem{lib-binomial}
``Gsl library,'' \url{http://www.gnu.org/software/gsl}.

\bibitem{simus-liopaul}
``Source code of the simulations,''
  \url{http://graal.ens-lyon.fr/~prenaud/colgen/}.

\end{thebibliography}

\end{document}